\documentclass[review]{elsarticle}
\usepackage{hyperref}
\journal{Neurocomputing}
\biboptions{sort&compress}

\usepackage{graphicx} 
\usepackage{bm}          
\usepackage{amsmath}
\usepackage{amssymb}
\usepackage{amsthm}
\usepackage[shortlabels]{enumitem} 
\usepackage{mdwlist}
\usepackage{xcolor}
\usepackage{subfigure}
\usepackage{epstopdf}

\newtheorem{definition}{Definition}
\newtheorem{assumption}{Assumption}
\newtheorem{proposition}{Proposition}
\newtheorem{theorem}{Theorem}
\newtheorem{lemma}{Lemma}
\newtheorem{remark}{Remark}

\def\conv{\operatorname{conv}} 
\def\dist{\operatorname{dist}}
\def\prj{\operatorname{prj}}
\def\uv{\operatorname{uv}}
\def\sgn{\operatorname{sgn}}
\def\s{\operatorname{s}}
\def\t{\operatorname{t}}
\def\c{\operatorname{c}}
\def\z{\operatorname{z}}
\def\T{\operatorname{T}}
\def\phii{\bm{\varphi}_i}
\def\y{\bm{y}}
\def\x{\bm{x}}
\def\u{\bm{u}}
\def\z{\bm{z}}

\newcommand{\DOI}[1]{doi: \href{https://doi.org/#1}{#1}}

\begin{document}

\begin{frontmatter}

\title{Achieving safe minimum circle circumnavigation around multiple targets: a dynamic compensation approach}

\author[UESTC]{Chao Wang}
\author[UESTC]{Yingjing Shi\corref{correspondingauthor}}
  \ead{shiyingjing@uestc.edu.cn}
\author[UESTC]{Rui Li}
\author[Song1,Song2]{Yongduan Song}

\cortext[correspondingauthor]{Corresponding author}
\address[UESTC]{School of Automation Engineering, University of Electronic Science and Technology of China, Chengdu 611731, PR China}
\address[Song1]{Key Laboratory of Dependable Service Computing in Cyber Physical Society of Ministry of Education, Chongqing University, Chongqing 400044, PR China}
\address[Song2]{School of Automation, Chongqing University, Chongqing 400044, PR China}

\begin{abstract}
  Minimum circle circumnavigation is proposed in this paper, which is of special value in target monitoring, capturing and/or attacking. In this paper, a safe minimum circle circumnavigation of multiple targets based on bearing measurements is studied. In contrast with the traditional circumnavigation problem, with the new pattern, one agent is able to enclose multiple targets along a minimum circle with the desired enclosing distance and tangential speed. To achieve the minimum circle circumnavigation, an algorithm including dynamic compensators and a control protocol is proposed, by which collision is avoided between the agent and the multiple targets during the whole moving process. Moreover, the control protocol developed for a single agent is further extended to the scenarios of multiple agents by adding a coordination mechanism into the tangential velocity term, which drives the agents to distribute evenly on each expected circular orbit with the same radius or different radius. Simulations results illustrate the effectiveness of the proposed methods.
\end{abstract}

\end{frontmatter}

\section{Introduction}

Due to the wide applications in applications in satellite clusters, monitoring, rescue, etc, considerable effort has been made to study the coordination control of the multi-agent system including consensus, formation, containment control and so on \cite{1,2,3,4}. In the area of cooperative control of multi-agent systems, circumnavigation is an important problem, which is to drive an agent/agents to enclose some targets or an area of interest.
A classic method dealing with such problem is that the agent first localize the targets, and then encloses the targets with the desired
radius and tangential speed.
Previous literature taking this classic approach includes, but is not limited to, \cite{5,6,7,8}.

According to the availability of target information, circumnavigation problems can be divided into full information-based and incomplete information-based circumnavigation. So far, full information-based circumnavigation problem has been widely investigated in literature, e.g., \cite{5,6,7,8,9,10}, where
the agent can directly obtain the position information of the target. However, such method is not applicable in some scenes where the target's position cannot be obtained directly. Therefore, the incomplete information-based circumnavigation problem has attracted increasing attention, see \cite{11,12,13,14,15,16,17,18,19,20,21,22,23,24,25}. Compared with the distance-only measurement adopted in \cite{11,12,13,14,15,16}, the bearing-only measurement is a passive measurement, which is favourable in some special scenes such as radio silence and has been investigated in literature, see e.g.,  \cite{17,18,19,20,21,22,23,24,25}.

In the early studies on circumnavigations with only bearing information, \cite{17} is a pioneering work, where an algorithm is proposed enabling an agent to enclose a target with the desired radius and tangential velocity.
The target's position is directly estimated by constructing the errors of the target position through orthogonal projection and the errors converge to zero with the movement of the agent. In \cite{18}, this algorithm is further extended for a multi-target scene, where multiple agents enclose multiple targets with the desired radius.

In recent studies, other key issues are considered while solving the localization and circumnavigation problem.
To closely detect or monitor the targets, one important problem is concerned with
 how to realize a circumnavigation pattern that enables an agent to enclose the targets closely.
In addition, how to design an algorithm to avoid collision with the targets is also worth investigating when enclosing unknown targets with only bearing measurements.

There have been some research efforts to investigate such localization and circumnavigation problems. Different from the previous circular circumnavigation, the authors of \cite{19} propose an elliptical circumnavigation pattern. Since  elliptical circumnavigation is closer than the traditional circle circumnavigation, it shows a better performance in harsh environments.
But collision avoidance with the targets is not considered and the distribution range of targets need to be known in advance. Although \cite{20} proposes a safe circle circumnavigation pattern, where the distance between the agent and the targets is larger than a safe distance, the pattern is not close enough and the distribution of initial positions of the estimates of the targets is strictly limited as well. Subsequently, a new scalar estimator is designed in \cite{21}, which provides a collision avoidance mechanism with a loose initial condition.
However, \cite{21} does not consider multi-target scenarios.
Inspired by \cite{19,21}, a safe convex-circumnavigation pattern is proposed in \cite{22} for multiple targets. The pattern realizes the closer circumnavigation than elliptical circumnavigation under loose initial conditions. Moreover, there is no collision between the agent and the convex hull of multiple targets.
However, the convex-circumnavigation pattern is not good enough in engineering practice to design control law. The reason is that the curvature of the orbit in the convex hull vertex region is large, thus it is difficult to control the agent to turn in this region. It should be noted that the control law is different from the control protocol. The speed of the agent is regulated with the control protocol and the force acting on the agent is regulated with the control law. With this consideration, a smooth circumnavigation pattern is expected because smooth orbitals require simpler accelerations than complex orbitals.

Inspired by the above analysis, in this paper, we consider the localization and safe minimum circle circumnavigation problem for multiple targets using bearing-only measurements,
where the agent can safely circumnavigate a minimum radius circle  covering all targets with the desired enclosing distance and tangential velocity.
To deal with the incomplete measurement information, dynamic compensators are first designed.
Based on the dynamic compensation approach proposed in this paper, a control protocol is designed to make the agent achieve the minimum circle circumnavigation.
To make the agent circumnavigate the targets closely enough, different from the control protocol proposed in \cite{18,19}, where the velocity of the agent is decomposed into the direction of the geometric center of the targets, we develop a new control protocol with the agent velocity pointing to the estimated minimum circle center and its orthogonal direction.
In addition, by the developed method, there is no collision between the agent and the multiple targets during the whole moving process.
Furthermore, the proposed method can be easily extended for the multiple agents scenarios, where an evenly spaced configuration along the desired minimum circle orbit is achieved.

The main contributions of this work are threefold. First, a minimum circle circumnavigation pattern is proposed. With this pattern, the agent can surround the targets more closely than the traditional circular circumnavigation, which can help the agent to detect the targets more thoroughly.
Second, a dynamic compensation based anti-collision mechanism is established without knowing the distribution range of the targets, by which the minimum circle circumnavigation is achieved, and no collision happens under the loose initial constrained conditions.
Third, compared with the convex-circumnavigation, under the minimum circle circumnavigation pattern, it is easier to design the control law and is more convenient to adjust the proposed algorithm for the multi-agent case.

The rest of this paper is structured as follows. In Section~\ref{section:prelimaries}, we introduce some mathematical preliminaries and formally propose the problem of localization and safe minimum circle circumnavigation of
multiple targets. In Section~\ref{section:algorithm}, we propose an algorithm to solve this problem and prove the convergence of the algorithm.
In Section~\ref{section:extension}, we extend the proposed algorithm to multiple targets scenarios.
Simulation results are shown in Section~\ref{section:simulation} and we demonstrate the conclusions and future works in Section~\ref{section:conclusion}.

\section{Preliminaries and problem statement}\label{section:prelimaries}

In this section, we first define some notations used throughout the paper, and then briefly review some closely related theories.

\subsection{Mathematical preliminaries and notations}

Suppose there are $n$ stationary targets with unknown positions and an agent moving in a 2D space. Both the targets and the agent are assumed to be modeled as points. Among them, the targets are represented by $\mu_{1}$, $\mu_{2}$,..., $\mu_{n}$.

In the beginning, some definitions of mathematical symbols are introduced.
\par

\begin{definition}\label{definition:1}
Several concepts need to be defined:
\begin{basedescript}{\desclabelstyle{\pushlabel}\desclabelwidth{6em}} 
  \setlength{\labelsep}{-0.5em} 

\item[$\boldsymbol{y}(t) \in \mathbb{R}^{2}$] the position of the agent

\item[$\boldsymbol{x}_{i} \in \mathbb{R}^{2}$] the position of the stationary target $\mu_{i}$

\item[$X \subset \mathbb{R}^{2}$] the set of positions of all of the targets, $X \triangleq \left\{\bm{x}_{i}\right\}_{i=1}^{n}$

\item[$\rho_{i}(t) \in \mathbb{R}_{\ge 0}$] the distance between the agent and the target $\mu_{i}$, $\rho_{i}(t) \triangleq \| \boldsymbol{x}_{i}-\boldsymbol{y}(t) \|$

\item[$\bm\varphi_{i}(t) \in \mathbb{U}$] the unit vector from the agent to the target $\mu_{i}$, $\bm{\varphi}_{i}(t) \triangleq [ \boldsymbol{x}_i(t) - \boldsymbol{y}(t) ] / \rho_i(t)$ $($if $\rho_{i}(t) > 0)$

\item[$\Phi(t) \subset \mathbb{U}$] the set of unit vectors from the agent to the targets at time $t$, $\Phi(t) \triangleq\left\{\boldsymbol{\varphi}_{i}(t) \right\}_{i=1}^{n}$

\item[$\bar{\boldsymbol{\varphi}}_{i}(t) \in \mathbb{U} $] $\boldsymbol{\varphi}_{i}(t)$ rotates $\pi / 2$ clockwise

\item[$\hat{\rho}_{i}(t) \in \mathbb{R} $] the distance between the agent and the estimated position of the target $\mu_{i}$ along the direction of $\bm{\varphi}_{i}(t)$

\item[$\tilde{\rho}_{i}(t) \in \mathbb{R} $] the estimation error of the position of the target $\mu_{i}$, $\tilde{\rho}_{i}(t) \triangleq \hat{\rho}_{i}(t)-\rho_{i}(t)$

\item[$\hat{\x}_{i}(t) \in \mathbb{R}^{2}$]  the estimated position of the target $\mu_{i}$, $\hat{\boldsymbol{x}}_{i}(t) \triangleq \boldsymbol{y}(t)+\hat{\rho}_{i}(t) \boldsymbol{\varphi}_{i}(t)$

\item[$\hat{X}(t) \subset \mathbb{R}^{2}$] the set of estimated positions of all targets, $\hat{X}(t) \triangleq \left\{\hat{\bm{x}}_{i}(t)\right\}_{i=1}^{n}$

\item[$\hat{\bm{c}}(t) \in \mathbb{R}^{2}$] the center coordinates of the minimum circle of the estimated positions of $n$ targets

\item[$\bm{c} \in \mathbb{R}^{2}$] the center coordinates of the minimum circle of the positions of $n$ targets

\item[$\hat{\rho}(t) \in \mathbb{R}_{\geq0}$] the distance between the center of the minimum circle of the estimated targets and the agent, $\hat{\rho}(t) \triangleq \|\hat{\boldsymbol{c}}(t) - \boldsymbol{y}(t)\|$

\item[$\rho(t) \in \mathbb{R}_{\geq0}$] the distance between the center of the minimum circle of the targets and the agent, $\rho(t) \triangleq \|\boldsymbol{c}(t) - \boldsymbol{y}(t)\|$

\item[$\hat{\boldsymbol{\varphi}}(t) \in \mathbb{U}$] the unit vector from the agent to the center of the minimum circle of the estimated positions of n targets, $\hat{\boldsymbol{\varphi}}(t)  \triangleq \left[\hat{\boldsymbol{c}}(t)-\boldsymbol{y}(t)\right] / \hat{\rho}(t)$ $($if $\hat{\rho}(t) > 0)$

\item[$\hat{\bar{\boldsymbol{\varphi}}}(t) \in \mathbb{U} $] $\hat{\boldsymbol{\varphi}}(t)$ rotates $\pi / 2$ clockwise

\item[$\boldsymbol{\varphi}(t) \in \mathbb{U}$] the unit vector from the agent to the center of the minimum circle of the positions of n targets, $\boldsymbol{\varphi}(t)  \triangleq \left[\boldsymbol{c}(t)-\boldsymbol{y}(t)\right] / \rho(t)$ $($if $\rho(t) > 0)$

\item[$\bar{\boldsymbol{\varphi}}(t) \in \mathbb{U} $] $\boldsymbol{\varphi}(t)$ rotates $\pi / 2$ clockwise

\item[$\alpha \in \mathbb{R}_{+}$] the desired tangential enclosing speed

\item[$r_{\s} \in \mathbb{R}_{+}$] the safety distance

\item[$\hat{r}_{\t}(t) \in \mathbb{R}_{\geq0}$] the radius of minimum circle of the estimated positions of $n$ targets

\item[${r_{\t}} \in \mathbb{R}_{+}$] the radius of minimum circle of the real positions of $n$ targets

\item[${d} \in \mathbb{R}_{+}$] the desired distance around the minimum circle

\item[$\hat{r}(t) \in \mathbb{R}_{+}$] the estimated radius of the desired orbit, $\hat{r}(t) \triangleq \hat{r}_{\t}(t)+d$

\item[$r \in \mathbb{R}_{+}$] the real radius of the desired orbit, $r \triangleq r_{\t}+d$

\end{basedescript}
\end{definition}
\noindent Here, $\mathbb{R}$ is the set of real numbers.
$\mathbb{R}_{\ge 0}$ is the set of non-negative real numbers.
$\mathbb{R}_+$ is the set of positive real numbers.
$\mathbb{R}^2$ is the set of 2D column vectors.
$\| \cdot \|$ is the Euclidian norm.
$\mathbb{U}$ is the set of unit vectors, $\mathbb{U} \triangleq \{ \boldsymbol{u} \in \mathbb{R}^2 \mid \| \boldsymbol{u} \| = 1 \}$.

\begin{remark}
  Some points need to be explained.
  \begin{enumerate}[(1)]
  \item $\hat{\rho}_{i}(t)$ is a variable that can be positive or negative.
  \item The minimum circle of a plane point set is unique, thus $\hat{r}_{\t}(t)$, $\hat{\bm{c}}(t)$ are unique at a certain time.
  \end{enumerate}
\end{remark}

In this paper, the concept of the convex hull is used to analyze the security of the agent motion, for which the related definitions are introduced.

\begin{definition}[Convex Hull of Finite Points Set]\label{definition:convexhull}\textsuperscript{\cite{22}}
  Suppose $P = \left\{\boldsymbol{p}_{i}\right\}_{i=1}^{m}$ is a subset of $\mathbb{R}^2$.
  We define $\conv(P)$ as the convex hull of $P$, such that
  \[ \conv(P) \triangleq \left\{ \sum_{i=1}^{m} k_{i} \boldsymbol{p}_{i} \mid k_{i} \ge 0,i=1,2,\dots,m, \sum_{i=1}^{m} k_{i}=1 \right\} .\]
\end{definition}

\begin{definition}[Relationship Between a Point and a Convex Set]\label{definition:relationship}\textsuperscript{\cite{22}}
  Suppose $Q \in \mathbb{R}^{2}$ is a set and $\boldsymbol{a} \in \mathbb{R}^{2}$ is a point. Define $\dist(\boldsymbol{a},Q)$ as the distance from $\boldsymbol{a}$ to Q, that is,
  $$\dist(\boldsymbol{a},Q) \triangleq \inf_{\boldsymbol{q} \in Q} \|\boldsymbol{a}-\boldsymbol{q}\|.$$
  If Q is compact and convex, then we can define $\prj(\boldsymbol{a},Q)$ as the minimum distance point in Q to $\boldsymbol{a}$, which satisfies that
  $$\|\prj(\boldsymbol{a},Q) - \boldsymbol{a}\| = \dist(\boldsymbol{a},Q).$$
  If $\dist(\boldsymbol{a},Q)>0,$ define $\uv(\boldsymbol{a},Q)$ as the unit vector pointing from $\boldsymbol{a}$ to Q, i.e.,
  $$\uv(\boldsymbol{a},Q) \triangleq \frac{\prj(\boldsymbol{a},Q)- \boldsymbol{a}}{\dist(\boldsymbol{a},Q)}.$$
\end{definition}
\begin{remark}
  In Definition~\ref{definition:relationship}, it should be noted that $\prj(\boldsymbol{a},Q)$ exists and  unique.
\end{remark}

With the concepts defined above, we give some definitions of the convex hull related to safe minimum circle circumnavigation.

\begin{definition}Define:
\begin{basedescript}{\desclabelstyle{\pushlabel}\desclabelwidth{6em}} 
  \setlength{\labelsep}{-0.5em} 
  \item[$X_{\c}\in\mathbb{R}^{2}$] the convex hull of the targets, $X_{\c}\triangleq\conv(X)$

  \item[$\hat{X}_{\c}(t)\in\mathbb{R}^{2}$] the convex hull of the estimated targets, $\hat{X}_{\c}(t)\triangleq\conv(\hat{X}(t))$

  \item[$D(t)\in\mathbb{R}_{\geq0}$] the distance from the agent to the convex hull of all of the targets, $D(t)\triangleq\dist(\y(t),X_{\c})$

  \item[$\hat{D}(t)\in\mathbb{R}_{\geq0}$] the distance from the agent to the convex hull of all of the targets, $\hat{D}(t)\triangleq\dist(\y(t),\hat{X}_{\c})$

  \item[$\bm{\psi}(t)\in\mathbb{R}^{2}$] the unit vector pointing from the agent to the convex hull of all of the targets, $\bm{\psi}(t)\triangleq\uv(\y(t),X_{\c})$ $($if $D(t)>0)$

  \item[$\bar{\bm{\psi}}(t)\in\mathbb{R}^{2}$] $\bm{\psi}(t)$ rotates $\pi/2$ clockwise

  \item[$\hat{\bm{\psi}}(t)\in\mathbb{R}^{2}$] the unit vector pointing from the agent to the convex hull of all of the targets, $\bm{\psi}(t)\triangleq\uv(\y(t),\hat{X}_{\c})$ $($if $\hat{D}(t)>0)$

  \item[$\hat{\bar{\bm{\psi}}}(t)\in\mathbb{R}^{2}$] $\hat{\bm{\psi}}(t)$ rotates $\pi/2$ clockwise

\end{basedescript}
\end{definition}

The relationships among some notations defined above are illustrated in  \figurename~\ref{figure:Relationship among variables.} and \figurename~\ref{figure:Minimum circle variables description.}.

\begin{figure}[!htbp]
  \centering
 \begin{minipage}[]{0.596\linewidth}
  \includegraphics[width=1\linewidth]{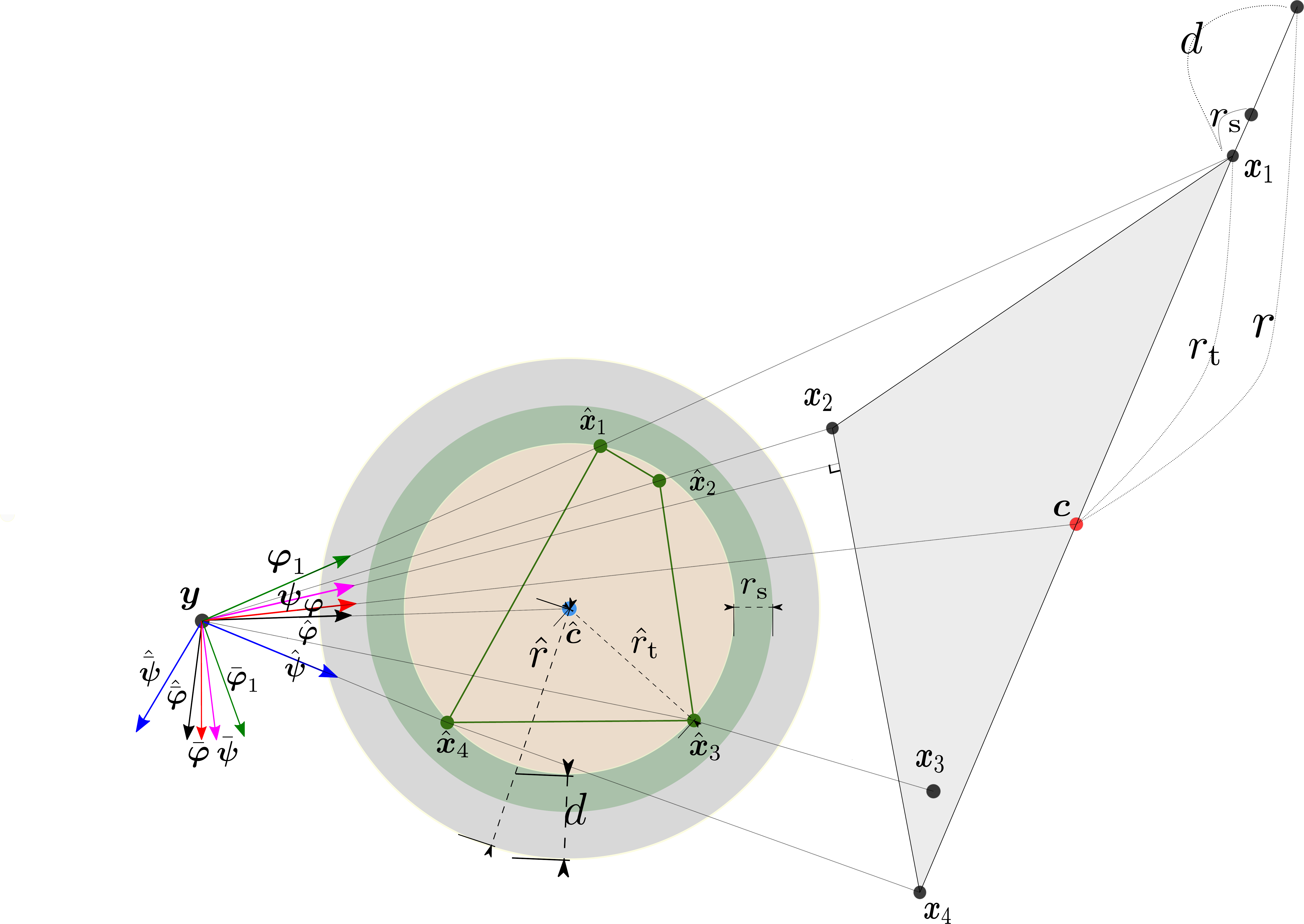}
   \caption{Relation among variables ($|X|=4$).}
   \label{figure:Relationship among variables.}
 \end{minipage}
 \hfill
 \centering
 \begin{minipage}[]{0.385\linewidth}
  \includegraphics[width=1\linewidth]{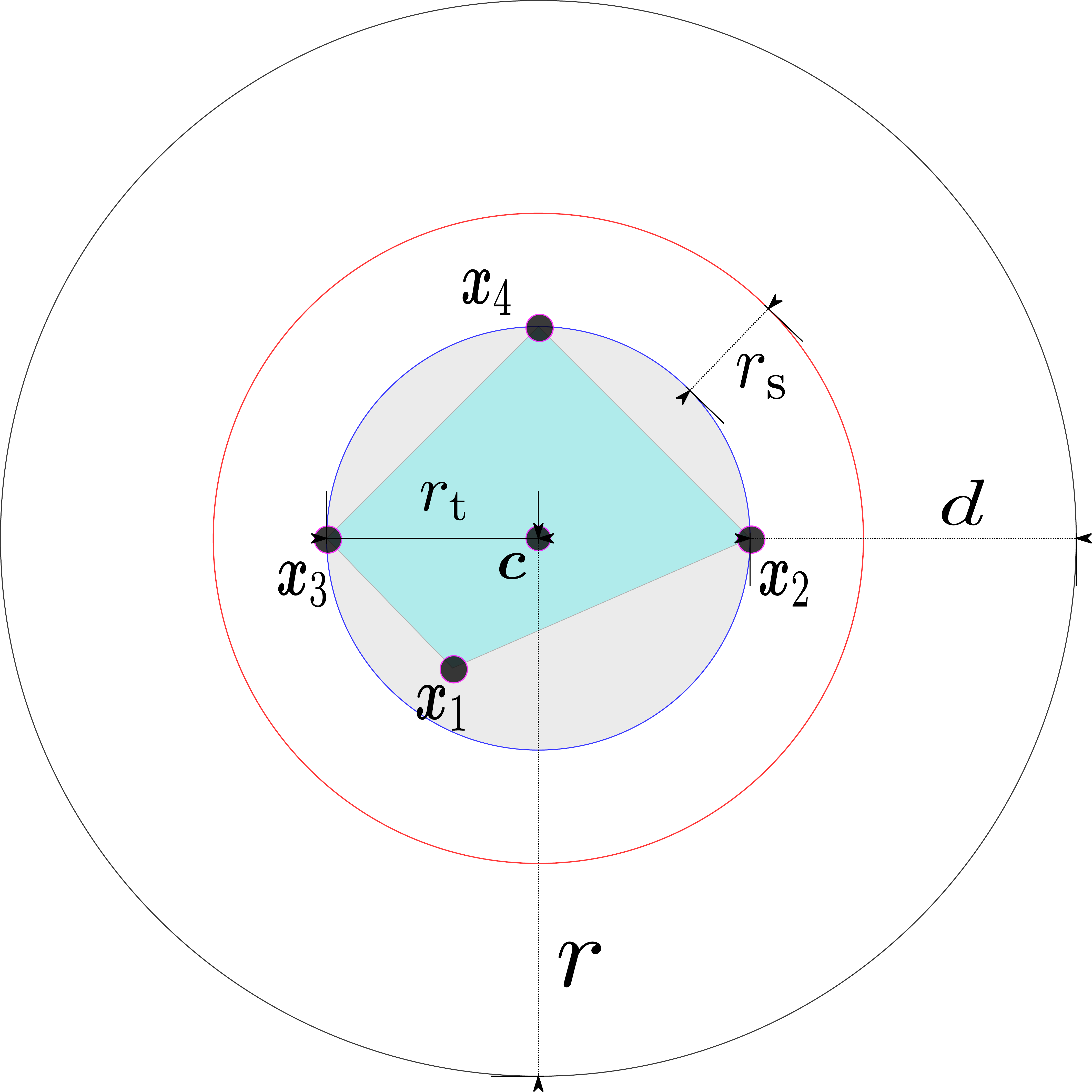}
  \caption{Minimum circle variables description.}
  \label{figure:Minimum circle variables description.}
 \end{minipage}
 \end{figure}

Then, some angle descriptions and comparison criteria of angle size are defined in advance for the convenience of the following description.
First, $\langle\bm{u},\bm{v}\rangle\in \mathbb{R}_{\geq0}$ is defined as the angle between $\bm{v}$ and $\bm{u}$. Then, suppose $V \triangleq \left\{\bm{v}_{i}\right\}_{i=1}^{m}$ is a subset of $\mathbb{U}.$
 $\bm{v}_r$ is defined as the rightmost unit vector of $V$ denoted by ruv(V), and $\bm{v}_l$ is defined as the leftmost unit vector of $V$ denoted by luv(V). Finally, we define the standard of angle size comparison. Note that the angle of $\bm{v}_i$ is defined as the angle when $\bm{v}_r$ is taken as the starting vector and rotated counterclockwise to $\bm{v}_i$.\par

Furthermore, we summarize some important results that will be used to analyze the convergence of dynamic compensators and control protocol.

\begin{proposition}\label{proposition:2}\textsuperscript{\cite{22}}
  Suppose $P \triangleq \left\{\bm{p}_{i}\right\}_{i=1}^{m}$ is a subset of $\mathbb{R}^{2}$. Let $P_{\c}\triangleq \conv(P)$, $\bm{a} \in  \mathbb{R}^{2}\setminus P_{\c}$, $\bm{v}_{i} \triangleq (\bm{p}_{i}-\bm{a})/\|\bm{p}_{i}-\bm{a}\|$ $(i=1,2,...,m)$, $V \triangleq \left\{\bm{v}_{i}\right\}_{i=1}^{m}$, $\bm{u} \triangleq \uv(\bm{a},P_{\c})$. Then we can assert that $\langle\bm{u},\bm{v}_{i}\rangle<\pi/2,$ for $i=1,2,...,m$.
\end{proposition}

\begin{proposition}\label{proposition:3}\textsuperscript{\cite{22}}
  Suppose $V \triangleq \left\{\bm{v}_{i}\right\}_{i=1}^{m}$ is a subset of $\mathbb{U}$. Let $P \triangleq \left\{k_{i}\bm{v}_{i}\right\}_{i=1}^{m}$, $Q \triangleq \left\{l_{i}\bm{v}_{i}\right\}_{i=1}^{m}$ be two sets where $k_{i}, l_{i} \in\mathbb{R}$, $0<k_{i}\leq l_{i}$ for i=1,2,...,m. Define $P_{\c} \triangleq \conv(P)$, $Q_{c} \triangleq \conv(Q)$, $L_{P} \triangleq \dist(\bm{0},P_{\c})$ and $L_{Q} \triangleq \dist(\bm{0},Q_{c})$. Then we can assert that
  $L_{P}\leq L_{Q}.$
\end{proposition}

\begin{proposition}\label{proposition:4}\textsuperscript{\cite{26}}
  Suppose $C \subset \mathbb{R}^{2}$ is a convex hull of a finite point set. Define a function
  \begin{center}
   $f: \mathbb{R}^{2}\backslash X_{c}\rightarrow \mathbb{R}_{+},$\par
   $\qquad \qquad \qquad \bm{a}\mapsto \dist(\bm{a},X_{\c}).$\par
  \end{center}
  Let $\bm{l} \triangleq \uv(\bm{a},C)$. Then we can assert that $f$ is differentiable and
  $\nabla f(\bm{a})=-\bm{l}.$
\end{proposition}

\begin{proposition}\label{proposition:5}\textsuperscript{\cite{27}}
  In a given point set, there are two or three points can always be found to determine the minimum covering circle of the point set. If there are three points on the minimum circle, then the minimum circle is the circumscribed circle of the three points; If there are only two points on the minimum circle, then the minimum circle is the circle with the diameter of the straight line segment between the two points.
\end{proposition}

\begin{proposition}\label{proposition:6}
  The center of the minimum circle of a convex hull composed of a set of plane points must be on the interior or boundary of the convex hull.
\end{proposition}
\begin{proof}
  With Proposition~\ref{proposition:5}, we know the center of the minimum circle is either the incenter of the triangle or the midpoint of the longest side of the triangle. Obviously, the center of the circle can be expressed linearly by the three vertices of the triangle, which are all points on the convex hull.
   According to the definition of the convex set, we know that the center of the circle is also a point on the convex hull. Hence, the conclusion of this proposition follows readily.
\end{proof}

\subsection{Problem statements}

Consider a multi-agent system containing multiple unknown static targets and an agent, the agent can only measure its own position as well as the bearing information from itself to each target, while the distance information from the agent to each target is unavailable. Assume that the agent is with single integrator dynamics, i.e., the velocity of the agent can be directly controlled. The purpose of this paper is to design an algorithm for an agent to locate the unknown targets and enclose them with a safe minimum circle. Concretely, the objectives of the paper include:\par
\begin{enumerate}[1)]
  \item An algorithm should be designed to locate the targets. That is, for all the targets, $\lim_{t \to +\infty }\tilde{\rho}_{i}(t)=0,$ $i=1,2,...,n.$
  \item  A control protocol is designed to make the agent stably encircle the minimum circle of all the targets with the desired distance and tangential speed. That is, $\lim_{t \to +\infty }\rho(t)=r$ and $\lim_{t \to +\infty } \bar{\boldsymbol{\varphi}}^{\mathrm{T}}(t) \dot{\y}(t) =\alpha.$
  \item Safety should be guaranteed. That is, the algorithm shall ensure that collision never happens.
\end{enumerate}

\begin{remark}
  The common definition of collision is that the agent does not collide with the targets. But the agent should avoid entering the area occupied by the targets to achieve real security in some application scenarios. Therefore, in this paper, we say collision happens at time $t'$ if $D(t')<r_{\s}$.
\end{remark}

\section{Main results}\label{section:algorithm}
\subsection{Proposed algorithm}
In this section, dynamic compensators are designed to compensate for the lack of measurement values of some state variables in the system, and then a control protocol is developed to achieve the safe minimum circle circumnavigation pattern through the action of the dynamic compensators.

First, the dynamics of the agent is described by the following nonlinear system
\begin{equation}\label{equation:state}
\left\{
\begin{aligned}
&\dot{\x}=\bm{f}(t,\x,\u), \\
&\z=\bm{h}(t,\x,\u).
\end{aligned}
\right.
\end{equation}
Here, $\x$, $\bm{u}$ and $\bm{z}$ are the state, control input and control output of the agent respectively.  $\bm{f}(\cdot)$ and $\bm{h}(\cdot)$ are unknown vector-valued functions. In fact, the speed of the agent is the control input, and the position of the agent is the system output. That is, $\bm{u}=\dot{\y}$ and $\bm{z}=\y$.

Then, in the multi-agent circumnavigation control system, we expect to take the relative position of the agent and the targets as the state of the system.  Compared with the cartesian coordinates,  in order to make full use of the known bearing information, we use the polar coordinates to describe the system state.  Therefore, we choose $\bm{\varphi}_{i}$ and $\rho_{i}$ as the state variables of the system in this paper. The state equation of this system is as follows
\begin{equation}\label{equation:state1}
\left\{
\begin{aligned}
&\dot{\bm{\varphi}}_{i}=f_{1}(t,\bm{\varphi}_{i},\rho_{i},\bm{u})=-\bar{\bm{\varphi}}_{i} \bar{\bm{\varphi}}_{i}^{\mathrm{T}} \dot{\y}/\rho_{i},\\
&\dot{\rho}_{i}=f_{2}(t,\bm{\varphi}_{i},\rho_{i},\bm{u})=-\bm{\varphi}^{\mathrm{T}}_{i}\dot{\y}.
\end{aligned}
\right.
\end{equation}
In this paper, the agent does not know the position of the targets, i.e., the information of state variable $\rho_{i}$ are unknown, which are required for the design of the control input.
Therefore, we need first to design dynamic compensators to compensate for the lack of the measured values of $\rho_{i}$.
The dynamic compensators are designed as follows
\begin{equation}\label{equation:estimators}
\left\{
\begin{aligned}
&\hat{\x}_{i}=\gamma_{1}(t,\x'_{i},\bm{z},\bm{u},\hat{\rho}_{i})=\y+\hat{\rho}_{i}\phii, \\
&\dot{\hat{\rho}}_{i}=\gamma_{2}(t,\x'_{i},\bm{z},\bm{u},\hat{\rho}_{i})=-\bm{\varphi}^{\mathrm{T}}_{i}\dot{\y}u(\hat{\rho}_{i}-h)+\sgn(\bar{\bm{\varphi}}^{\mathrm{T}}_{i}\dot{\y})(\bar{\bm{\varphi}}^{\mathrm{T}}_{i}\dot{\y}+
\hat{\rho}_{i}\bar{\bm{\varphi}}^{\mathrm{T}}_{i}\dot{\bm{\varphi}}_{i}).
\end{aligned}
\right.
\end{equation}
Here, $\hat{\rho}_{i}$ and $\hat{\x}_{i}$ are the dynamic compensator state  variables and output respectively. In addition, $\sgn(\cdot)$ is the sign function, $u(\cdot)$ is the step function and $h$ is a positive constant.

\begin{remark}
In the design of dynamic compensators, some points need to be explained.
  \begin{enumerate}[(1)]
  \item $\x'_{i}$ in~\eqref{equation:estimators} is different from $\x_{i}$, where $\x_{i}$ refers to the full information of the targets and $\x'_{i}$ refers to the incomplete information of the targets.
      The function of $\x'_{i}$ is to provide a reference point for the agent to obtain the bearing information to the targets without using the coordinate information of the targets.
  \item Since $\bm{\varphi}_{i}$ and $\y$ can be measured, the term $\sgn(\bar{\bm{\varphi}}^{\mathrm{T}}_{i}\dot{\y})\bar{\bm{\varphi}}^{\mathrm{T}}_{i}\dot{\bm{\varphi}}_{i}$
      is known to us until time $t$ if we view the derivative $\dot{\bm{\varphi}}$ as left derivative. Therefore, $\dot{\hat{\rho}}_{i}$ can be solved.
  \end{enumerate}

\end{remark}


Under the action of dynamic compensators, the control input (control protocol) is designed as follows
\begin{equation}\label{equation:protocol}
\u(t)=\dot{\y}(t)=k[\hat{\rho}(t)-\hat{r}(t)]\hat{\bm{\varphi}}(t)+\alpha u(\hat{\rho}(t)-\hat{r}_{\t}(t)-r_{\s})\hat{\bar{\bm{\varphi}}}(t).
\end{equation}

\begin{remark}
  As long as $\hat{\x}_{i}$ is determined, the estimated minimum circle is unique.  So only know $\y$ and $\hat{\x}_{i}$, one can get $\hat{r}_{\t}$, $\hat{r}$, $\hat{\rho}$, $\hat{\bm{\varphi}}$ and $\hat{\bar{\bm{\varphi}}}$.
\end{remark}

Finally, to ensure that the agent can safely enclose the minimum circle of the targets, we need the following assumptions.

\begin{assumption}\label{assumption:D0>r}
  $ r_{\s}\leq D(0)$.
\end{assumption}

\begin{assumption}\label{assumption:d>r_s}
  $h<r_{\s}<d$.
\end{assumption}

\begin{assumption}\label{assumption:0<rhohat0<rho0}
  $h\leq \hat{\rho}_{i}(0) \leq \rho_{i}(0)$ for $i=1,2,...,n $.
\end{assumption}
\noindent These assumptions are quite reasonable and weak. Assumption~\ref{assumption:D0>r} is to ensure that collision does not happen in the beginning. Assumption~\ref{assumption:d>r_s} means that the desired enclosing distance is larger than the safety distance and the safety distance is larger than a small positive constant. Assumption~\ref{assumption:0<rhohat0<rho0} is a constraint on the initial value of each target's estimate and is loose.

\subsection{Stability Analysis}\label{section:stability}
In this section, we prove the stability and convergence property of the proposed algorithms. We first verify the safety can be guaranteed, and then prove that the agent eventually performs safe minimum circle circumnavigation around all of the targets successfully.

\subsubsection{Safety}
In this subsection, our goal is to confirm that collision does not happen in the whole time domain by using the proposed algorithm. To better explain that safety performance is guaranteed, the concept of the convex hull is applied. However, it should be noted that in real implementation, it is not necessary to construct the convex hull. \par
Denote

\begin{equation*}
v_{i}(t) \triangleq \bm{\varphi}^{\mathrm{T}}_{i}(t)\dot{\y}(t),\quad \bar{v}_{i}(t) \triangleq \bar{\bm{\varphi}}^{\mathrm{T}}_{i}(t)\dot{\y}(t),\quad \omega_{i}(t) \triangleq \frac{\bar{v}_{i}(t)}{\rho_{i}(t)}.
\end{equation*}
Omitting ``$(t)$'' for simplicity, it follows that

\begin{align}
  \dot{\bm{\varphi}}_{i}
  &=\notag
  \frac{\mathrm{d}}{\mathrm{d} t} \frac{\x_{i}-\y}{\left\|\x_{i}-\y\right\|}
  =
  \frac{\left(I-\bm{\varphi}_{i}\bm{\varphi}_{i}^{\mathrm{T}}\right)\left(\dot{\x}_{i}-\dot{\y}\right)}{\left\|\x_{i}-\y\right\|}
  \\&=
  \frac{-\bar{\bm{\varphi}}_{i} \bar{\bm{\varphi}}_{i}^{\mathrm{T}} \dot{\y}}{\left\|\x_{i}-\y\right\|}
  =
  -\frac{\bar{v}_{i}}{\rho_{i}} \bar{\bm{\varphi}}_{i}
  =\label{equation:4}
  -\omega_{i} \bar{\bm{\varphi}}_{i},
\end{align}
\begin{equation}\label{equation:5}
  \dot{\hat{\x}}_{i}=\frac{\mathrm{d}}{\mathrm{d} t} (\y+\hat{\rho}_{i}\phii)=\dot{\y}+\dot{\hat{\rho}}_{i}\phii+\hat{\rho}_{i}\dot{\phii},
\end{equation}
\begin{align}
  \dot{\hat{\rho}}_i
  &=\notag
  -\bm{\varphi}^{\mathrm{T}}_{i}\dot{\y}u(\hat{\rho}_{i}-h)
  +\sgn(\bar{\bm{\varphi}}^{\mathrm{T}}_{i}\dot{\y})(\bar{\bm{\varphi}}^{\mathrm{T}}_{i}\dot{\y}+
\hat{\rho}_{i}\bar{\bm{\varphi}}^{\mathrm{T}}_{i}\dot{\bm{\varphi}}_{i})
  \\&=\notag
  - v_{i}u(\hat{\rho}_{i}-h) + | \bar{v}_i | ( 1 - \frac{\hat{\rho}_i}{\rho_i} )
  \\&=\label{equation:6}
  - v_{i}u(\hat{\rho}_{i}-h) - | \bar{v}_i | \frac{\tilde{\rho}_i}{\rho_i} ,
\end{align}
\begin{equation}\label{equation:7}
  \dot{\rho}_{i}=\frac{\mathrm{d}}{\mathrm{d} t}\|\x_{i}-\y\|=\phii^{\mathrm{T}}(\dot{\x}_{i}-\dot{\y})=-v_{i},
\end{equation}
and
\begin{align}
  \dot{\tilde{\rho}}_{i}
  &=\notag
  \dot{\hat{\rho}}_{i}-\dot{\rho}_{i}
  =v_{i}[1-u(\hat{\rho}_{i}-h)]-\frac{|\bar{v}_{i}|}{\rho_{i}}\tilde{\rho_{i}}
  \\&=\label{equation:8}
  v_{i}[1-u(\hat{\rho}_{i}-h)]-|\omega_{i}|\tilde{\rho}_{i}.
\end{align}

To ensure that the collision will not occur, we first give the important angle relations.
\begin{lemma}\label{lemma:2}
  If $D(t)>0$ and $\hat{D}(t)>0$, then it is true that
  $\bm{\psi}^{\mathrm{T}}(t)\bm{\varphi}(t)>0$, $\bm{\psi}^{\mathrm{T}}(t)\hat{\bm{\varphi}}(t)>0.$
\end{lemma}

\begin{proof}
  With Proposition~\ref{proposition:2}, we know that $\hat{\bm{\psi}}^{\mathrm{T}}(t)\bm{\varphi}_{i}(t)>0$ and $\bm{\psi}^{\mathrm{T}}(t)\bm{\varphi}_{i}(t)>0$ hold at any time $t$ for $i=1,2,...,n$. With Proposition~\ref{proposition:6}, we learn that $\bm{\varphi}(t)$ and $\hat{\bm{\varphi}}(t)$ must be inside $\bm{\varphi}_r$ and $\bm{\varphi}_l$. So $\langle\hat{\bm{\varphi}}(t),\bm{\psi}(t)\rangle \leq \max_{1\leq i\leq n}\langle\bm{\varphi}_{i}(t),\bm{\psi}(t)\rangle <\pi/2$, and $\langle\bm{\varphi}(t),\bm{\psi}(t)\rangle<\pi/2$ holds as well, which imply that $\bm{\psi}^{\mathrm{T}}(t)\bm{\varphi}(t)>0$ and $\bm{\psi}^{\mathrm{T}}(t)\hat{\bm{\varphi}}(t)>0.$ Hence, the conclusion of this lemma follows readily.
\end{proof}

Then, based on the above lemma and under some restricted conditions, the following lemma is given to ensure the safety.
\begin{lemma}\label{lemma:4}
  Suppose Assumptions~\ref{assumption:D0>r},~\ref{assumption:d>r_s},~\ref{assumption:0<rhohat0<rho0} hold. Under the dynamic compensators~\eqref{equation:estimators} and the control protocol ~\eqref{equation:protocol}, if $0 < \hat{\rho}_{i}(t) \leq \rho_{i}(t),$ for $i=1,2,...,n$, then we have
  $$D(t)\geq r_{\s}.$$
\end{lemma}

\begin{proof}
   Considering Proposition~\ref{proposition:3} and $0 < \hat{\rho}_{i}(t) \leq \rho_{i}(t)$, we know $\hat{D}(t)\leq D(t)$. If $0<D(t)<r_{\s}$, we have $\hat{D}(t)\leq D(t)<r_{\s}$. It follows that $\hat{\rho}(t)<\hat{r}_{\t}(t)+r_{\s}$. Thus $\dot{\y}(t)=k[\hat{\rho}(t)-\hat{r}(t)]\hat{\bm{\varphi}}(t)$.
   Since $D(t)>0$, $\bm{\psi}(t)$ is definable. Define a function \par
  \begin{center}
     $f: \mathbb{R}^{2}\backslash X_{c}\rightarrow \mathbb{R}_{+},$\par
   $\qquad \qquad \qquad \bm{a}\mapsto \dist(\bm{a},X_{c}).$\par
  \end{center}
   In view of Proposition~\ref{proposition:4}, we know $\nabla f(\y(t))=-\bm{\psi}(t).$ Using the chain rule, we have
  $$\dot{D}(t)=[\nabla f(\y(t))]^{\mathrm{T}}\dot{\y}(t)=-\bm{\psi}^{\mathrm{T}}(t)\dot{\y}(t)=-k[\hat{\rho}(t)-\hat{r}(t)]\bm{\psi}^{\mathrm{T}}(t)\hat{\bm{\varphi}}(t).$$
  According to  Lemma~\ref{lemma:2}, from Assumption~\ref{assumption:d>r_s}, we know that $\bm{\psi}^{\mathrm{T}}(t)\hat{\bm{\varphi}}(t)>0$ and $\hat{\rho}(t)<\hat{r}(t)$ since $d>r_{\s}$, which implies that $\dot{D}(t)>0$ if $D(t)<r_{\s}$. With Assumption~\ref{assumption:D0>r}, $D(0)\geq r_{\s}$, then it follows that $D(t)\geq r_{\s}$. This completes the proof.

\end{proof}
Finally, the following two lemmas prove that the given constraint conditions are satisfied.
\begin{lemma}\label{lemma:5}
  Suppose Assumptions~\ref{assumption:D0>r},~\ref{assumption:d>r_s},~\ref{assumption:0<rhohat0<rho0} hold. Under the dynamic compensators~\eqref{equation:estimators} and the control protocol~\eqref{equation:protocol}, if $\tilde{\rho}_{i}(t)\leq0$, then we have $\hat{\rho}_{i}(t)\geq h$.
\end{lemma}

\begin{proof}
  It follows from~\eqref{equation:6} that $\dot{\hat{\rho}}_{i}(t)=- | \bar{v}_i(t) | \tilde{\rho}_i(t)/\rho_i(t)\geq0$ if $\hat{\rho}_{i}(t)< h$.
  Thus it implies that $\dot{\hat{\rho}}_{i}(t)\geq0$ if $\hat{\rho}_{i}(t)< h$. With Assumption~\ref{assumption:0<rhohat0<rho0}, we can conclude that $\hat{\rho}_{i}(t)\geq h$.
\end{proof}

\begin{lemma}\label{lemma:6}
  Suppose Assumptions~\ref{assumption:D0>r},~\ref{assumption:d>r_s},~\ref{assumption:0<rhohat0<rho0} hold. Under the dynamic compensators~\eqref{equation:estimators} and the control protocol~\eqref{equation:protocol}, we have
  $h \leq \hat{\rho}_{i}(t) \leq \rho_{i}(t),$ for $i=1,2,...,n$.
\end{lemma}

\begin{proof}

 First, we prove that $\hat{\rho}_{i}(t) \leq \rho_{i}(t),$ that is, $\tilde{\rho}_{i}(t)\leq0$, which is performed by three steps.
 \par
  \emph{Step 1.}
 Suppose there exists $t_1>0$ such that $\tilde{\rho}_{i}(t_{1})=0$. Define the zero set $A\triangleq\{t\mid\tilde{\rho}_{i}(t)=0\}$. $A$ is a closed set because $\tilde{\rho}_{i}(t)$ is continuous and set $\{0\}$ is closed. Define $t_2\triangleq\min\{A\}$.
 With Assumption~\ref{assumption:0<rhohat0<rho0},
 when $t\in[0,t_{2}]$, we have $\tilde{\rho}_{i}(t)\leq0$. With Lemma~\ref{lemma:5},
 $h \leq \hat{\rho}_{i}(t) \leq \rho_{i}(t)$ hold. Then with Lemma~\ref{lemma:4},
 we have $\rho_{i}(t_2)=\hat{\rho}_{i}(t_2)\geq D(t_2)\geq r_{\s}> h$.
 \par
  \emph{Step 2.}
 We prove that there is $\hat{\rho}_{i}(t)> h$ when $t > t_2$.
 Apply the reduction to absurdity. Suppose there exists $t_3>t_2$ such that $\hat{\rho}_{i}(t_{3})=h$. Define the zero set $B\triangleq\{t\mid\hat{\rho}_{i}(t)=h, t\geq t_2\}$. $B$ is also a closed set. Define $t_4\triangleq\min\{B\}$.\par
 When $t\in[t_2,t_{4}]$, $\hat{\rho}_{i}(t)\geq h$. With~\eqref{equation:8}, we get $\dot{\tilde{\rho}}_{i}(t)=-|\omega_{i}|\tilde{\rho}_{i}$. Solving the differential equation, $\tilde{\rho}_i (t) = \tilde{\rho}_i(t_2)\exp\left(- \!\! \int_{t_2}^{t} |\omega_{i}(\tau)| \mathrm{d}\tau \right)=0$, it turns out that $h \leq \hat{\rho}_{i}(t) = \rho_{i}(t).$ With Lemma~\ref{lemma:4}, $\rho_{i}(t_4)=\hat{\rho}_{i}(t_4)\geq D(t_4)\geq r_{\s}> h$ still holds, so the contradiction appears. Therefore,  when $t>t_2$, $\hat{\rho}_{i}(t)=h$ is impossible, which also implies that $\hat{\rho}_{i}(t)> h$ and $\tilde{\rho}_{i}(t)=0$ hold when $t > t_2$.
 \par
  \emph{Step 3.}
  Based on the foregoing analysis, when $\tilde{\rho}_{i}(t_2)=0$, then $\hat{\rho}_{i}(t)\geq h$
  holds for $t > t_2$. With~\eqref{equation:8}, we have $\tilde{\rho}_{i}(t)=0$ when $t\geq t_2$. Since $\tilde{\rho}_{i}(t)$ is continuous and $\tilde{\rho}_{i}(0) < 0$, one can conclude that $\tilde{\rho}_{i}(t)\leq0$.\par

 The remaining part of the proof proceeds by showing that $\hat{\rho}_{i}(t)\geq h.$ In fact, from the conclusion of the previous step and Lemma~\ref{lemma:5}, we can directly get $ \hat{\rho}_{i}(t)\geq h$.\par
 Hence, the conclusion of this lemma follows readily.
\end{proof}

\begin{remark}
  Since $h \leq \hat{\rho}_{i}(t) \leq \rho_{i}(t)$ always holds, with lemma~\ref{lemma:4}, we know $D(t)\geq r_{\s}$ is always established. So the agent does not collide with the convex hull of a group of targets in the moving process of the agent, which guarantees real security.
\end{remark}

\begin{lemma}
  Suppose Assumptions~\ref{assumption:D0>r},~\ref{assumption:d>r_s},~\ref{assumption:0<rhohat0<rho0} hold. Under the dynamic compensators~\eqref{equation:estimators} and the control protocol~\eqref{equation:protocol}, we have $\hat{D}(t)>0.$
\end{lemma}

\begin{proof}
  With Lemma~\ref{lemma:4} and Lemma~\ref{lemma:6}, we know $D(t)>r_{\s}$, thus $\langle \bm{\varphi}_{l}(t),\bm{\varphi}_{r}(t)\rangle < \pi$. Let  $\hat{\x}'_{i}(t)\triangleq \y(t)+h\bm{\varphi}_{i}(t),$ then $\hat{D}'(t)=h\cos(\langle \bm{\varphi}_{l}(t),\bm{\varphi}_{r}(t)\rangle /2)
  >0$. Since $\rho_{i}(t)\geq\hat{\rho}_{i}(t)\geq h$, with Proposition~\ref{proposition:3}, we know $\hat{D}(t)\geq \hat{D}'(t)>0.$ Hence, the conclusion of this lemma follows readily.
\end{proof}

\begin{remark}
  Based on the above analysis, we know that $D(t)>0$ and $\hat{D}(t)>0$. With Proposition~\ref{proposition:6}, we have $\hat{\rho}(t)\geq\hat{D}(t)$ and $\rho(t)\geq D(t)$. Therefore, $\bm{\varphi}_{i}(t),$ $\hat{\bm{\varphi}}(t),$ $\bm{\varphi}(t),$ $\hat{\bm{\psi}}(t),$ $\bm{\psi}(t)$ are definable.
\end{remark}

\subsubsection{Circumnavigation}\label{section:circumnavigation}
In this subsection, we prove that the agent finally performs minimum circle circumnavigation around all of the targets. The proof is divided into two steps. First, we prove that the agent will circumnavigate the minimum circle of some points with the desired enclosing distance $d$ and the tangential speed $\alpha$. Then, we prove that these points are the real targets.

Firstly, we present two lemmas that will be used in the convergence analysis of the proposed algorithm.

\begin{lemma}\label{lemma:7}
  Suppose Assumptions~\ref{assumption:D0>r},~\ref{assumption:d>r_s},~\ref{assumption:0<rhohat0<rho0} hold. Under the dynamic compensators~\eqref{equation:estimators} and the control protocol ~\eqref{equation:protocol}, if for some $i_0$, there is
  $$\lim_{t \to +\infty }\tilde{\rho}_{i_0}(t)\neq0,$$
  then $\bm\varphi_{i_0}(t)$ must converge to some unit vector.
\end{lemma}

\begin{proof}
  With Lemma~\ref{lemma:6},
  considering~\eqref{equation:8}, we have $\dot{\tilde{\rho}}_{i}(t)=-|\omega_{i}|\tilde{\rho}_{i}\geq0$. Then
  $$\tilde{\rho}_i (t) = \tilde{\rho}_i(0)e^{-\!\! \int_{0}^{t} | \omega_i (\tau) | \mathrm{d} \tau}, \quad t\geq0. $$
   It means that $\tilde{\rho}_{i}(t)$ is monotone and bounded, thus it must converge. If for some $i_0$, there is $\lim_{t \to +\infty }\tilde{\rho}_{i_0}(t)\neq0$, it is obviously that $\!\! \int_{0}^{+\infty} | \omega_{i_0} (\tau) | \mathrm{d} \tau <+\infty$, thus $\!\! \int_{0}^{+\infty}  \omega_{i_0} (\tau)  \mathrm{d} \tau $ converges.
  By using complex number to represent the vectors, that is, $\varphi_{i_0}(t)$ denotes $\boldsymbol{\varphi}_{i_0}(t)$ and $- j \varphi_{i_0}(t)$ denotes $\bar{\boldsymbol{\varphi}}_{i_0}(t)$, with~\eqref{equation:4}, we have
  $$\dot{\varphi}_{i_0}(t)=j\omega_{i_0}(t)\varphi_{i_0}(t).$$
  Solving this differential equation, we have
  $$\varphi_{i_0}(t)=\varphi_{i_0}(0)e^{j\!\! \int_{0}^{t}  \omega_{i_0} (\tau)  \mathrm{d} \tau}.$$
  Since $\!\! \int_{0}^{+\infty}  \omega_{i_0} (\tau)  \mathrm{d} \tau $ converges, $\varphi_{i_0}(t)$ converges and then $\bm{\varphi}_{i_0}(t)$ converges.
   This completes the proof.
\end{proof}

\begin{lemma}\label{lemma:8}
  Suppose Assumptions~\ref{assumption:D0>r},~\ref{assumption:d>r_s},~\ref{assumption:0<rhohat0<rho0} hold. Under the dynamic compensators~\eqref{equation:estimators} and the control protocol~\eqref{equation:protocol}, each $\hat{\x}_{i}(t)$ must converge to some point.
\end{lemma}

\begin{proof}
  With Lemma~\ref{lemma:7}, we know $\tilde{\rho}_{i}(t)$ must converge. Note that $\hat{\x}_{i}(t)=\x_{i}+\tilde{\rho}_{i}(t)\bm{\varphi}_{i}(t).$ If $\lim_{t \to +\infty }\tilde{\rho}_{i}(t)=0$, $\hat{\x}_{i}(t)$ converges to $\x_{i}(t).$ If $\lim_{t \to +\infty }\tilde{\rho}_{i}(t)\neq0,$ then with Lemma~\ref{lemma:7} again, we know $\bm{\varphi}_{i}(t)$ converges to some unit vector. Thus each $\hat{\x}_{i}(t)$ converges. Hence, the conclusion of this lemma follows readily.
\end{proof}
Subsequently, we prove that the agent will circumnavigate the minimum circle of some points with the desired enclosing distance and tangential velocity.

\begin{theorem}\label{theorem:3}
  If there exists a point set $Z = \{ \boldsymbol{z}_1 , \boldsymbol{z}_2 , \cdots , \boldsymbol{z}_n \} \subset \mathbb{R}^2$ such that $\lim_{t \to + \infty} \hat{\boldsymbol{x}}_i (t) = \boldsymbol{z}_i$, we shall know
  $$\lim_{t \to +\infty }\hat{\bm{c}}(t)=\bm{c}_{\z},\quad \lim_{t \to +\infty }\hat{r}_{\t}(t)=r_{\z}.$$
  Here, $\bm{c}_{\z}$ is the center of the minimum circle of $Z$, and $r_{\z}$ is the radius of the minimum circle of $Z.$
\end{theorem}

\begin{proof}
  Since $\lim_{t \to +\infty }\hat{\x}_{i}(t)=\boldsymbol{z}_i$, for any $\varepsilon>0$, there exists $t_{1}>0$ such that $\|\hat{\x}_{i}(t)-\boldsymbol{z}_i\|<\varepsilon,$ for $t>t_{1}.$  Thus when $t>t_1$, with Proposition~\ref{proposition:5}, we know that the center of the minimum circle can be divided into the following two cases.\par
  \medskip
  \emph{Case 1.}
  The center of the circle is the midpoint of the longest side. We know $\lim_{t \to +\infty }\hat{\bm{c}}(t)=\lim_{t \to +\infty }(\hat{\x}_{h}(t)+\hat{\x}_{j}(t))/2$, where $\hat{\x}_{h}(t)$, $\hat{\x}_{j}(t)$ are two points used to determine the estimated minimum circle. Similarly, $\bm{c}_{\z}=(\boldsymbol{z}_{h}+\boldsymbol{z}_{j})/2$, then
  \begin{eqnarray} \notag
  \begin{aligned}
  &
  \|\hat{\bm{c}}(t)-\bm{c}_{\z}\|=\frac{1}{2}\|\hat{\x}_{h}(t)+\hat{\x}_{j}(t)-\boldsymbol{z}_{h}-\boldsymbol{z}_{j}\|
  \\&\leq\frac{1}{2}(\|\hat{\x}_{h}(t)-\boldsymbol{z}_{h}\|+\|\hat{\x}_{j}(t)-\boldsymbol{z}_{j}\|)<\varepsilon.
  \end{aligned}
  \end{eqnarray}
  Therefore, $\|\hat{\bm{c}}(t)-\bm{c}_{\z}\|<\varepsilon$ for $t>t_{1}$, which implies that $\lim_{t \to + \infty }\hat{\bm{c}}(t)=\bm{c}_{\z}.$ It is obviously that $\lim_{t \to +\infty }\hat{r}_{\t}(t)=r_{\z}.$\par
  \medskip
  \emph{Case 2.}
  The center of the circle is the incenter of a triangle. According to the definition of the incenter of a triangle, $\lim_{t \to +\infty }\|\hat{\bm{c}}(t)-\hat{\x}_{k}(t)\|=\lim_{t \to +\infty }\|\hat{\bm{c}}(t)-\hat{\x}_{l}(t)\|=\lim_{t \to +\infty }\|\hat{\bm{c}}(t)-\hat{\x}_{p}(t)\|$, where $\hat{\x}_{k}(t)$, $\hat{\x}_{l}(t)$, $\hat{\x}_{p}(t)$ are three points used to determine the estimated minimum circle. Then,
  \begin{eqnarray} \notag
  \begin{aligned}
  &
  \|\hat{\bm{c}}(t)-\boldsymbol{z}_{k}\|-\|\hat{\x}_{k}(t)-\boldsymbol{z}_{k}\|
  \\&\leq\|\hat{\x}_{k}(t)-\boldsymbol{z}_{k}-(\hat{\bm{c}}(t)-\boldsymbol{z}_{k})\|=\|\hat{\x}_{k}(t)-\hat{\bm{c}}(t)\|
  \\&
    \leq\|\hat{\bm{c}}(t)-\boldsymbol{z}_{k}\|+\|\hat{\x}_{k}(t)-\boldsymbol{z}_{k}\|.
  \end{aligned}
  \end{eqnarray}
  According to $\lim_{t \to +\infty }\|\hat{\x}_{k}(t)-\boldsymbol{z}_{k}\|=0$, we know $\lim_{t \to +\infty }\|\hat{\bm{c}}(t)-\boldsymbol{z}_{k}\|=\lim_{t \to +\infty }\|\hat{\bm{c}}(t)-\hat{\x}_{k}(t)\|$. Similarly, $\lim_{t \to + \infty }\|\hat{\bm{c}}(t)-\boldsymbol{z}_{l}\|=\lim_{t \to +\infty }\|\hat{\bm{c}}(t)-\hat{\x}_{l}(t)\|$, $\lim_{t \to +\infty }\|\hat{\bm{c}}(t)-\boldsymbol{z}_{p}\|=\lim_{t \to +\infty }\|\hat{\bm{c}}(t)-\hat{\x}_{p}(t)\|.$ So
  $$\lim_{t \to +\infty }\|\hat{\bm{c}}(t)-\boldsymbol{z}_{k}\|=\lim_{t \to +\infty }\|\hat{\bm{c}}(t)-\boldsymbol{z}_{l}\|=\lim_{t \to +\infty }\|\hat{\bm{c}}(t)-\boldsymbol{z}_{p}\|.$$ Thus $\hat{\bm{c}}(t)$ is also the incenter of the triangle by $\boldsymbol{z}_{k}$, $\boldsymbol{z}_{l}$ and $\boldsymbol{z}_{p}$, which implies that $\hat{\bm{c}}(t)$ is also the center of the minimum circle of $Z$. Therefore, $\lim_{t \to +\infty }\hat{\bm{c}}(t)=\bm{c}_{\z}.$ Then, it is obviously that $\lim_{t \to +\infty }\hat{r}_{\t}(t)=r_{\z}.$\par
  Therefore, the conclusion of this lemma follows readily.
\end{proof}

\begin{theorem}\label{theorem:4}
  Suppose Assumptions~\ref{assumption:D0>r},~\ref{assumption:d>r_s},~\ref{assumption:0<rhohat0<rho0} hold. Under the dynamic compensators~\eqref{equation:estimators} and the control protocol~\eqref{equation:protocol}, the agent circumnavigates the minimum circle of $Z$ with the desired distance $d$ and tangential speed $\alpha$.
\end{theorem}

\begin{proof}
  Define $l_{\z}\triangleq r_{\z}+d,$ $\rho_{\z}(t)\triangleq \|\y(t)-\bm{c}_{\z}\|,$ $\bm{\varphi}_{\z}(t)\triangleq [ \bm{c}_{\z} - \boldsymbol{y}(t) ] / \rho_{\z}(t)$ and $\bar{\bm{\varphi}}_{\z}(t)$ denotes $\bm{\varphi}_{\z}(t)$ as it rotates $\pi/2$ clockwise.
  Let $e(t)\triangleq \hat{r}(t)-l_{\z},$ $w(t)\triangleq \hat{\rho}(t)-\rho_{\z}(t),$
  $\bm{\delta}(t)\triangleq \hat{\bm{\varphi}}(t)-\bm{\varphi}_{\z}(t)$ and $\bar{\bm{\delta}}(t)\triangleq \hat{\bar{\bm{\varphi}}}(t)-\bar{\bm{\varphi}}_{\z}(t).$ With Theorem~\ref{theorem:3}, we know $\lim_{t \to +\infty }\hat{r}(t)=l_{\z}.$ Obviously when $\lim_{t \to +\infty }\hat{\bm{c}}(t)=\bm{c}_{\z}$, we have
  \begin{equation}\label{equation:9}
     \lim_{t \to +\infty }e(t)=0,\ \lim_{t \to +\infty }w(t)=0,\ \lim_{t \to +\infty }\bm{\delta}(t)=\bm{0},\ \lim_{t \to +\infty }\bar{\bm{\delta}}(t)=\bm{0}.
  \end{equation}
  \noindent
  Similar to Proposition~\ref{proposition:4}, one can conclude that
  \begin{equation*}
  \dot{\rho}_{\z}(t)=-\bm{\varphi}^{\mathrm{T}}_{\z}(t)\dot{\y}(t)=-k\rho_{\z}(t)+kl_{\z}+g(t),
  \end{equation*}
  where $g(t)\triangleq k(e(t)-w(t))+k(l_{\z}+e(t)-w(t)-\rho_{\z}(t))\bm{\varphi}^{\mathrm{T}}_{\z}(t)\bm{\delta}(t)-\alpha u(\hat{\rho}(t)-\hat{r}_{\t}(t)-r_{\s})\bm{\varphi}^{\mathrm{T}}_{\z}(t)\bar{\bm{\delta}}(t).$
  With~\eqref{equation:9}, $\lim_{t \to +\infty }g(t)=0.$ Thus we have $\lim_{t \to +\infty }\rho_{\z}(t)=l_{\z}$, which also means that the distance between the agent and the minimum circle of $Z$ converges to $d$.\par
  Note that
  \begin{equation*}
    \bar{\bm{\varphi}}_{\z}^{\mathrm{T}}(t)\dot{\y}(t)=\alpha u(\hat{\rho}(t)-\hat{r}_{\t}(t)-r_{\s})+p(t),
  \end{equation*}
  where $p(t)\triangleq k[\hat{\rho}(t)-\hat{r}(t)]\bar{\bm{\varphi}}^{\mathrm{T}}_{\z}(t)\bm{\delta}(t)+\alpha u(\hat{\rho}(t)-\hat{r}_{\t}(t)-r_{\s})\bar{\bm{\varphi}}^{\mathrm{T}}_{\z}(t)\bar{\bm{\delta}}(t).$
  With~\eqref{equation:9}, we know $\lim_{t \to +\infty }\hat{\rho}(t)=\lim_{t \to +\infty }\rho_{\z}(t)=l_{\z}=\lim_{t \to +\infty }\hat{r}(t)=\lim_{t \to +\infty }\hat{r}_{\t}(t)+d$. Considering $d>r_{\s}$, there exists $t_{2}>0$, such that $\hat{\rho}(t)-\hat{r}_{\t}(t)-r_{\s}>0$ for $t>t_{2}$. Therefore, when $t>t_{2}$, $u(\hat{\rho}(t)-\hat{r}_{\t}(t)-r_{\s})=1$, and then one gets  $\bar{\bm{\varphi}}^{\mathrm{T}}_{\z}(t)\dot{\y}(t)=\alpha+p(t)$.
  Recalling $\lim_{t \to +\infty }p(t)=0$, we have $\lim_{t \to +\infty }\bar{\bm{\varphi}}^{\mathrm{T}}_{\z}(t)\dot{\y}(t)=\alpha$.
  It implies that the agent encircles the minimum circle of $Z$ with the desired tangential speed $\alpha$.
  The proof is then completed.
\end{proof}

Finally, we prove that the outputs of the dynamic compensator shall converge to the real targets and the agent performs the safe minimum circle circumnavigation around the real targets.

\begin{theorem}
  Suppose Assumptions~\ref{assumption:D0>r},~\ref{assumption:d>r_s},~\ref{assumption:0<rhohat0<rho0} hold. Under the dynamic compensators~\eqref{equation:estimators} and the control protocol~\eqref{equation:protocol}, the agent circumnavigates the minimum circle of the real targets with the desired distance $d$ and tangential speed $\alpha$.
\end{theorem}

\begin{proof}
    The proof is divided into two steps.
\par
  \emph{Step 1.}
  We prove that the localization tasks are achieved.
  \\
  \noindent
  Suppose there exists $i_0$, such that $\lim_{t \to +\infty }\tilde{\rho}_{i_0}(t)\neq 0$. Then with Lemma~\ref{lemma:7}, we know $\bm{\varphi}_{i_0}(t)$ must converge. With Theorem~\ref{theorem:4}, it is obvious that the agent must circumnavigate the minimum circle of $Z$ with the desired enclosing distance and tangential speed. Thus $\bm{\varphi}_{i_0}(t)$ is not converging. Contradiction appears. Therefore, we have $\lim_{t \to +\infty }\tilde{\rho}_{i}(t)=0,$ for $i=1,2,...,n$.
  \par
  \emph{Step 2.}
  We prove that the agent circumnavigates the minimum circle of the real targets with the desired distance $d$ and tangential speed $\alpha$.
  \\
  \noindent With step 1, we have $\lim_{t \to + \infty} \hat{\boldsymbol{x}}_i (t) = \boldsymbol{x}_i$, thus $X=Z.$ With Theorem~\ref{theorem:4}, we know that the conclusion of this theorem obviously holds.
\end{proof}

\section{Extension to multiple agents}\label{section:extension}
In this section, we will show the aforementioned algorithm can be extended for the multiple agents case. Inspired by \cite{24,25}, for the multiple agents' case, it is required that the agents not only converge to the desired shaped orbit but also maintain an equiangular spaced formation around the targets.

However, different from \cite{24,25}, this paper studies a multi-target circumnavigation problem. We do not know the direction of the agent pointing to the desired orbital center.
Therefore, we apply the direction pointing to the center of the minimum circle of the estimated targets instead. We define $\theta_{i}$ as the angle of $\hat{\bm{\varphi}}^{\hat{\c}_{i}}_{i},$ where $\hat{\bm{\varphi}}^{\hat{\c}_{i}}_{i}$  represents the unit vector pointing from agent $i$ to the minimum circle center estimated by agent $i$. The reference axis of $\theta_{i}$ is globally fixed in this paper.

 Convert the variables to local versions by adding the subscript $i$. To be more practical, it is assumed that the agents cannot detect each other unless their distance is less than a threshold distance $M$. We denote the neighbor set of agent $i$ as $\mathcal{N}_{i}= \{j\mid\|\y_{j}-\y_{i}\|\leq M, j\neq i\}$.
The bearing angle of agent $j$ from agent $i$ as $\theta^{j}_{i}$ being the angle of $(\y_{j}-\y_{i}), j\in \mathcal{N}_{i}$ and wrap $\{\theta_{i}-\theta^{j}_{i}, j\in \mathcal{N}_{i}\}$ to the interval $\left[-\pi,\pi \right)$. Then let $i^{+}= \{j\mid\max\{\theta_{i}-\theta^{j}_{i}\}, j\in\mathcal{N}_{i}\}$ and $i^{-}= \{j\mid\min\{\theta_{i}-\theta^{j}_{i}\}, j\in\mathcal{N}_{i}\}$.
It is worth pointing out that $i^+$ and $i^{-}$ can be determined only using bearings without communication.
Suppose that positions of $\hat{\c}_{i^{+}}$ and $\hat{\c}_{i^{-}}$ are available to agent $i$ from its neighboring agent $i^{+}$ and $i^{-}$ via communication respectively.
In addition, as long as the neighbor agents of agent $i$ are within its detection range, the agent $i$ can obtain the location of its neighbor agents.
Then we define

\[\sigma^{+}_{i}=\left\{\begin{array}{ll}

\arccos(\hat{\bm{\varphi}}^{\hat{\c}_{i^{+}} \mathrm{T}}_{i^{+}}\hat{\bm{\varphi}}^{\hat{\c}_{i}}_{i})&\text{if $\mathcal{N}^{+}_{i} \neq \emptyset$},\\

\pi&

\text{if $\mathcal{N}^{+}_{i} = \emptyset$},

\end{array}\right.\]
and
\[\sigma^{-}_{i}=\left\{\begin{array}{ll}

\arccos(\hat{\bm{\varphi}}^{\hat{\c}_{i^{-}} \mathrm{T}}_{i^{-}}\hat{\bm{\varphi}}^{\hat{\c}_{i}}_{i})&\text{if $\mathcal{N}^{-}_{i} \neq \emptyset$},\\

\pi&

\text{if $\mathcal{N}^{-}_{i} = \emptyset$},

\end{array}\right.\]
\noindent
where $\mathcal{N}^{+}_{i}=\{j\in\mathcal{N}_{i}\mid \theta_{i}-\theta^{j}_{i}\in(0,\pi)\}$ and $\mathcal{N}^{-}_{i}=\{j\in\mathcal{N}_{i}\mid \theta_{i}-\theta^{j}_{i}\in(-\pi,0)\}$.

To achieve circumnavigation of multiple agents for the targets at equiangular spaced formation on the same circular orbit, in fact, only the tangential velocity term in ~\eqref{equation:protocol} needs to be modified. Thus, we define the new control protocol as
\begin{equation}\label{equation:newprotocol}
\dot{\y}_{i}(t)=k[\hat{\rho}^{\hat{\c}_{i}}_{i}(t)-\hat{r}_{i}(t)]\hat{\bm{\varphi}}^{\hat{\c}_{i}}_{i}(t)+f(\sigma^{+}_{i}/\sigma^{-}_{i})\alpha u(\hat{\rho}^{\hat{\c}_{i}}_{i}(t)-\hat{r}_{i,\t}(t)-r_{\s})\hat{\bar{\bm{\varphi}}}^{\hat{\c}_{i}}_{i}(t),
\end{equation}

\noindent where $\hat{\rho}^{\hat{\c}_{i}}_{i}(t)$ represents the distance between the agent $i$ and the center of the minimum circle estimated by it, $\hat{r}_{i}(t)$ represents the desired orbit radius estimated by agent $i$, $\hat{r}_{i,\t}(t)$ is the minimum circle radius estimated by agent $i$, and $f : \left(0,+\infty \right]\rightarrow\left(0,1 \right]$ is any strictly increasing continuous function. For instance, we can let $f(x)=1-\exp(-x)$.

Define the angle of $-\hat{\bm{\varphi}}^{\hat{\c}_{i}}_{i}$ as $\phi_{i}=\theta_{i}+\pi$ mod $2\pi$. For convenience, we label agents at time $t$ according to their positions in a counterclockwise radial order around the targets, i.e., $j>k\Leftrightarrow\phi_{j}>\phi_{k}.$ We also define $\delta\phi_{i}=\phi_{i+1}-\phi_{i}$ for $i\in\{1,...,n-1\}$ $(n\geq3)$ and define $\delta\phi_{n}=\phi_{1}+2\pi-\phi_{n}.$ Under this labeling, $\delta\phi_{i}\in\left[0,2\pi \right)$ for all $i$.

\begin{proposition}
  Suppose Assumptions~\ref{assumption:D0>r},~\ref{assumption:d>r_s},~\ref{assumption:0<rhohat0<rho0} hold. Under the dynamic compensators~\eqref{equation:estimators} and the new control protocol~\eqref{equation:newprotocol}, the equilibrium point given by $\{\hat{\rho}^{\hat{\c}_{i}}_{i}=r, \hat{\c}_{i}=\c, \delta\phi_{i}=2\pi/n,\forall i\}$ is asymptotically stable if $M\geq2r\sin(\pi/n)$.
\end{proposition}

\begin{proof}
  Based on the analysis in Section~\ref{section:algorithm}, because $f(\sigma^{+}_{i}/\sigma^{-}_{i})\alpha \in (0,\alpha]$, then for all $i$, $\hat{\rho}^{\hat{\c}_{i}}_{i}\rightarrow r$ and $\hat{\c}_{i}=\c$ as $t\rightarrow \infty.$ The security of the agents is also guaranteed.
  The condition $M\geq2r\sin(\pi/n)$ guarantees that neither $\mathcal{N}^{+}_{i}$ nor $\mathcal{N}^{-}_{i}$ is empty around the equilibrium point.
  Because the minimum circle is strictly convex, it can be easily proved that $i^{+}=i+1$ and $i^{-}=i-1$ around the equilibrium point.
  Therefore, when
  $\hat{\rho}^{\hat{\c}_{i}}_{i} = r$ and $\hat{\c}_{i}=\c$ for all $i$, we have(see \figurename~\ref{figure:Angle.}) $\sigma^{+}_{i}=\delta\phi_{i}$ and $\sigma^{-}_{i}=\delta\phi_{i-1}$.
  With the new control protocol~\eqref{equation:newprotocol}, we know that the tangential velocity of each agent is $f(\sigma^{+}_{i}/\sigma^{-}_{i})\alpha$. Then, it is obtained that $\dot{\theta}_{i}=\dot{\phi}_{i}=f(\sigma^{+}_{i}/\sigma^{-}_{i})\alpha/r,$ resulting that $\frac{\mathrm{d}}{\mathrm{d} t} \delta\phi_{i}=[f(\delta\phi_{i+1}/\delta\phi_{i})-f(\delta\phi_{i}/\delta\phi_{i-1})]\alpha/r$.
  Define $\overline{\T}=\{j\mid\delta\phi_{j}=\max\delta\phi_{i}\}$ and $\underline{\T}=\{j\mid\delta\phi_{j}=\min\delta\phi_{i}\}$. When $\overline{\T}\neq\underline{\T}$, according to the fact that $f(x)$ is strictly increasing continuous function, we have
    $\frac{\mathrm{d}}{\mathrm{d} t}\frac{\Sigma_{j\in\overline{\T}} \delta\phi_{j}}{|\overline{\T}|}<0$, and
    $\frac{\mathrm{d}}{\mathrm{d} t}\frac{\Sigma_{j\in\underline{\T}} \delta\phi_{j}}{|\underline{\T}|}>0$.
  So we conclude that the equilibrium point defined by $\{\hat{\rho}^{\hat{\c}_{i}}_{i}=r, \hat{\c}_{i}=\c, \delta\phi_{i}=2\pi/n,\forall i\}$ is asymptotically stable.
\end{proof}

\begin{figure}[!htb]
  \centering
  \includegraphics[width=0.4\linewidth]{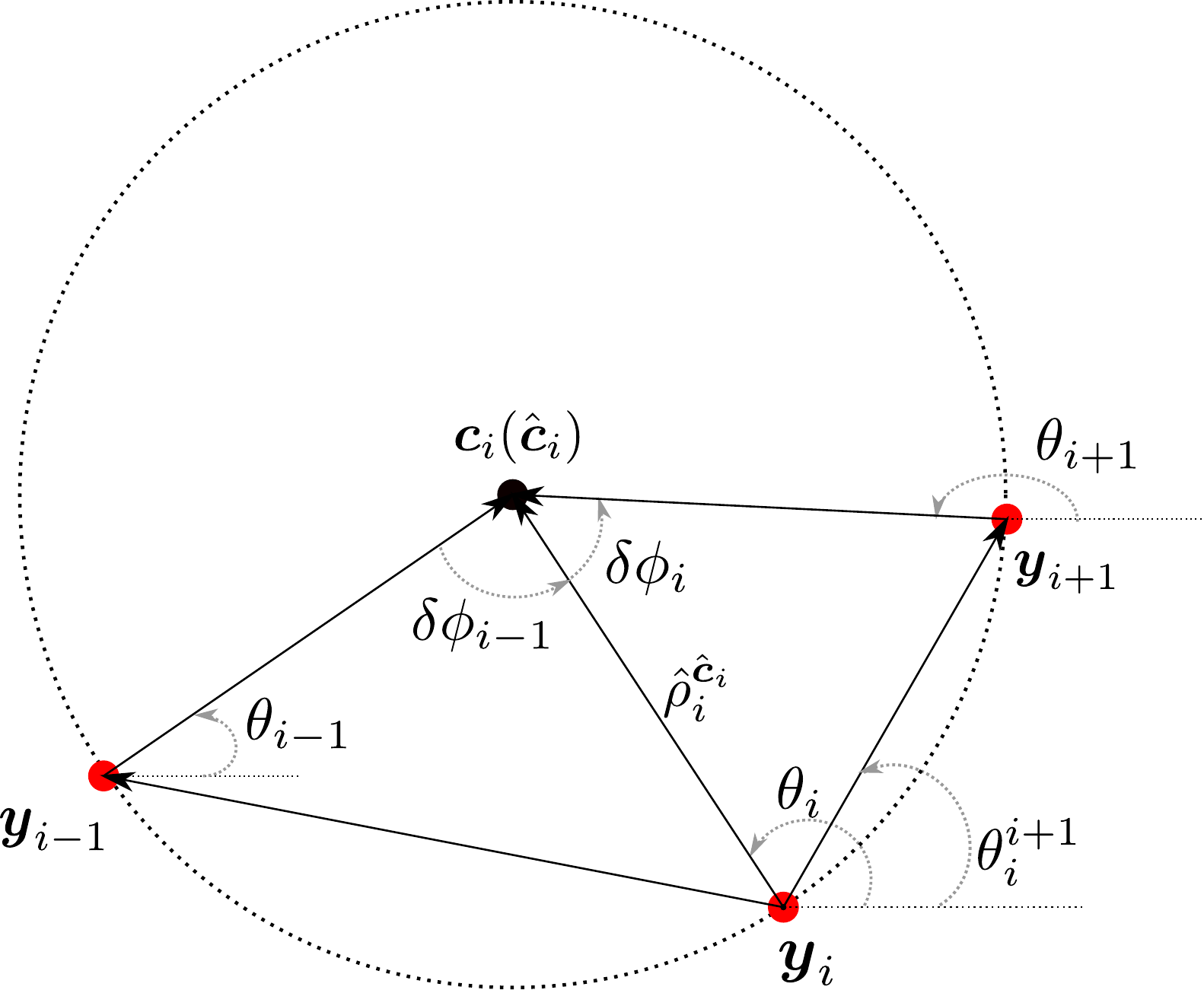}
   \caption{Relationship among variables when $\hat{\rho}^{\hat{\c}_{i}}_{i} = r$ and $\hat{\c}_{i}=\c$.}
   \label{figure:Angle.}
\end{figure}

\begin{remark}
  For the multiple agents' case, it is required that the agents not
only converge to the different desired circle orbit but also maintain an
equiangular spaced formation around the targets. Similarly, we just need to multiply the tangential velocity term in~\eqref{equation:newprotocol} by their respective desired radius.
\end{remark}
\section{Simulation}\label{section:simulation}
In this section, a series of simulation results for the safe minimum circle circumnavigation algorithm of four targets are presented. \par
\medskip
  \emph{Case 1.} For the case of an agent,
we set the coefficient $k=5$, the tangential speed $\alpha=5$, the safe distance $r_{\s}=0.3$, the desired distance $d=0.4$ and $h=0.1$. The initial values of the estimates are all set as $\hat{\rho}_{i}(0)=0.4$ $(i=1, 2, 3, 4)$. The positions of targets are $\x_{1}=[-2,0]^{\mathrm{T}},$ $\x_{2}=[4,5]^{\mathrm{T}},$ $\x_{3}=[2,0]^{\mathrm{T}},$ $\x_{4}=[1,1]^{\mathrm{T}}.$
The initial position of the agent is $\y(0)=[8,0]^{\mathrm{T}}$.\par
\figurename{}s~\ref{figure:The trajectories of the agent and the estimated targets.}--\ref{figure:Estimation errors of the targets' positions.} show that the proposed dynamic compensators and control protocol can achieve localization and safe minimum circle circumnavigation for multiple targets. \figurename~\ref{figure:The trajectories of the agent and the estimated targets.} describes the trajectories of the agent and the dynamic compensator outputs in 2D space.
It is shown in the figure that the agent circumnavigates the minimum circle of the targets with the desired enclosing distance.
\figurename~\ref{figure:Distance between the agent and the estimated minimum circle} describes the distance between the agent and the real minimum circle, which
shows that the distance eventually converges to the desired enclosing distance $d$.  \figurename~\ref{figure:Tangent velocity of agent} describes the image of $\bar{\bm{\varphi}}^{\mathrm{T}}(t)\dot{\y}(t)$, which means that the tangential velocity of the agent enclosing the minimum circle will eventually converge to the desired velocity $\alpha$.
 \figurename~\ref{figure:Estimation errors of the targets' positions.} depicts the estimation error $\tilde{\rho}_{i}(t)$ and shows that all the errors converge to 0, which means that the dynamic compensator outputs will converge to the real target's positions.

To sum up,  it can be concluded that there is no collision in the whole circumnavigation process, and the dynamic compensators and control protocol eventually converge. Therefore, the safe minimum circle circumnavigation is successfully performed by the proposed algorithm.

\begin{figure}[!htb]
  \centering
  \includegraphics[width=0.78\linewidth]{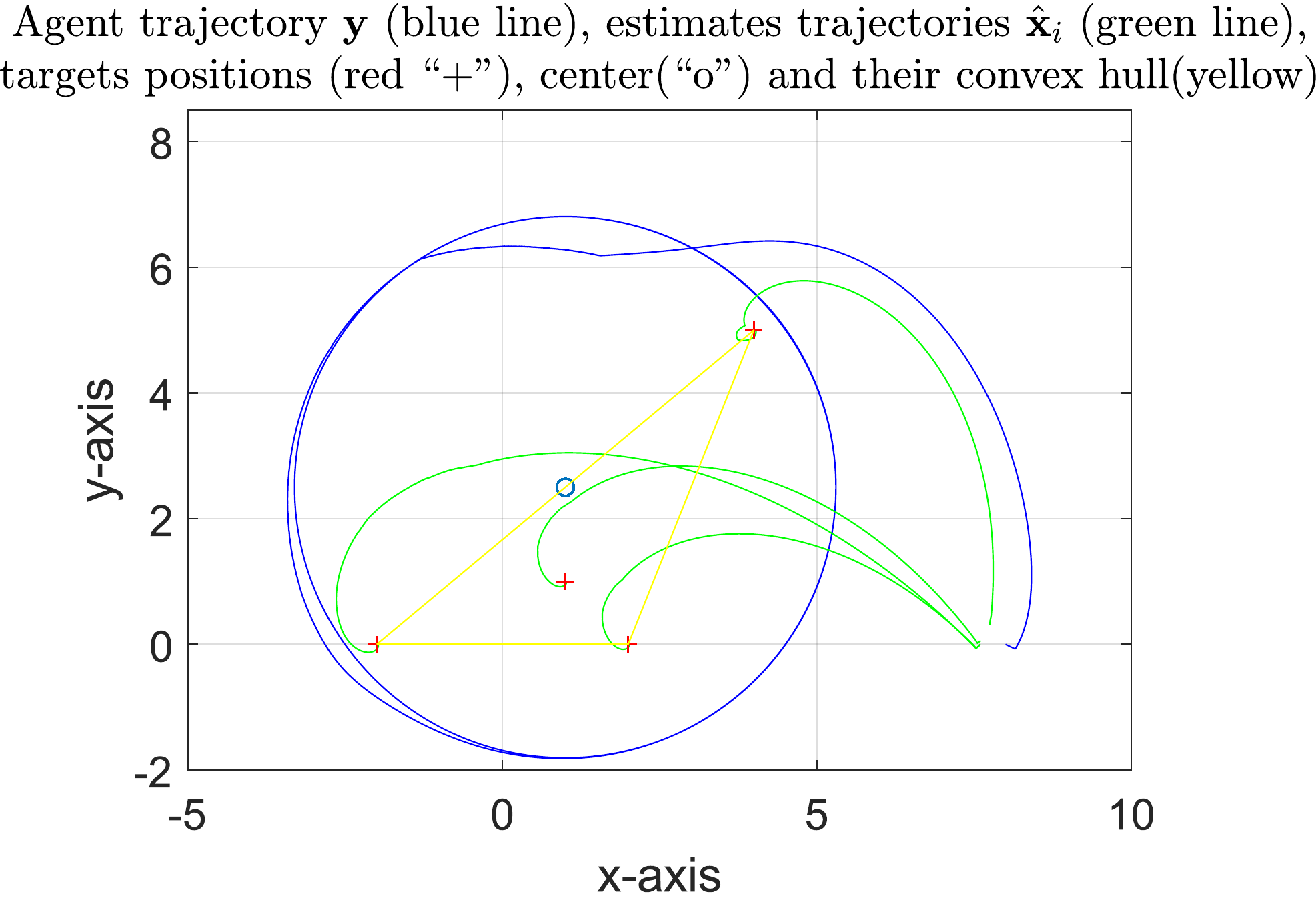}
   \caption{The trajectories of the agent and the dynamic compensator outputs.}
   \label{figure:The trajectories of the agent and the estimated targets.}
\end{figure}

\begin{figure}[!htbp]
  \centering
  \begin{minipage}[]{0.45\linewidth}
  \centering
  \includegraphics[width=1\linewidth]{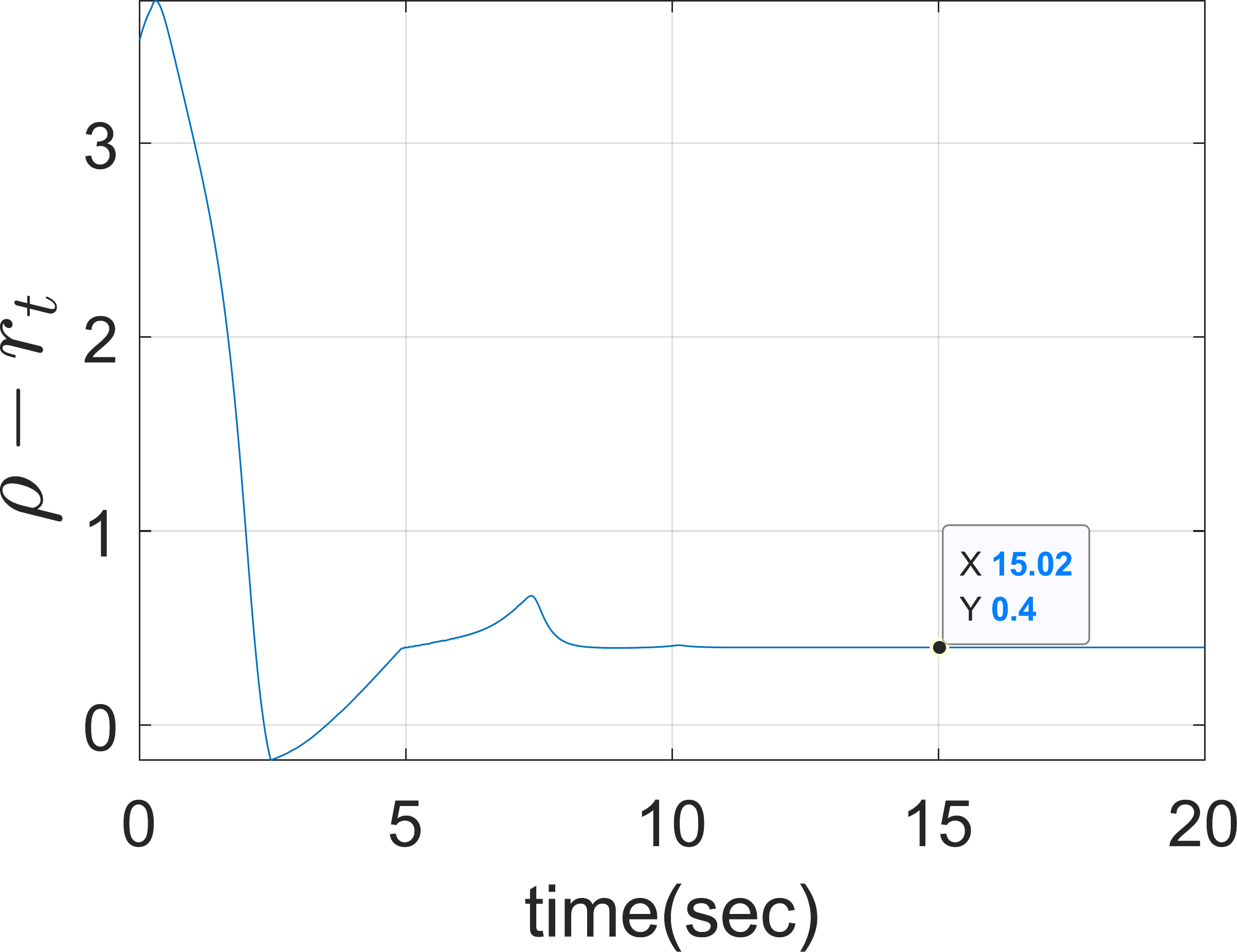}
  \caption{Distance between the agent and the real minimum circle.}
  \label{figure:Distance between the agent and the estimated minimum circle}
  \end{minipage}
  \hfill
  \begin{minipage}[]{0.45\linewidth}
  \centering
  \includegraphics[width=1\linewidth]{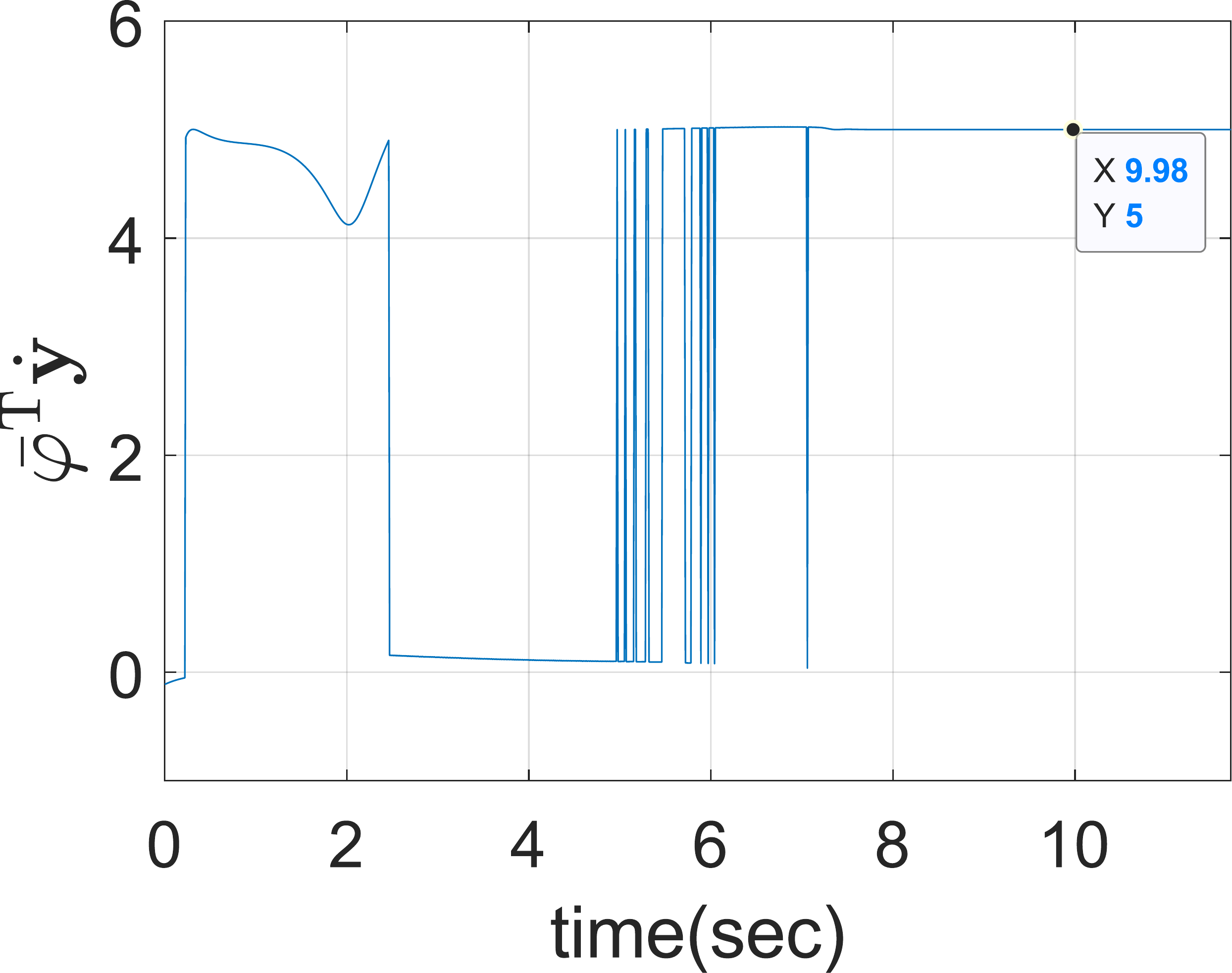}
  \caption{The tangential velocity of the agent.}
  \label{figure:Tangent velocity of agent}
  \end{minipage}
\end{figure}

\begin{figure}[!htb]
  \centering
  \includegraphics[width=0.6\linewidth]{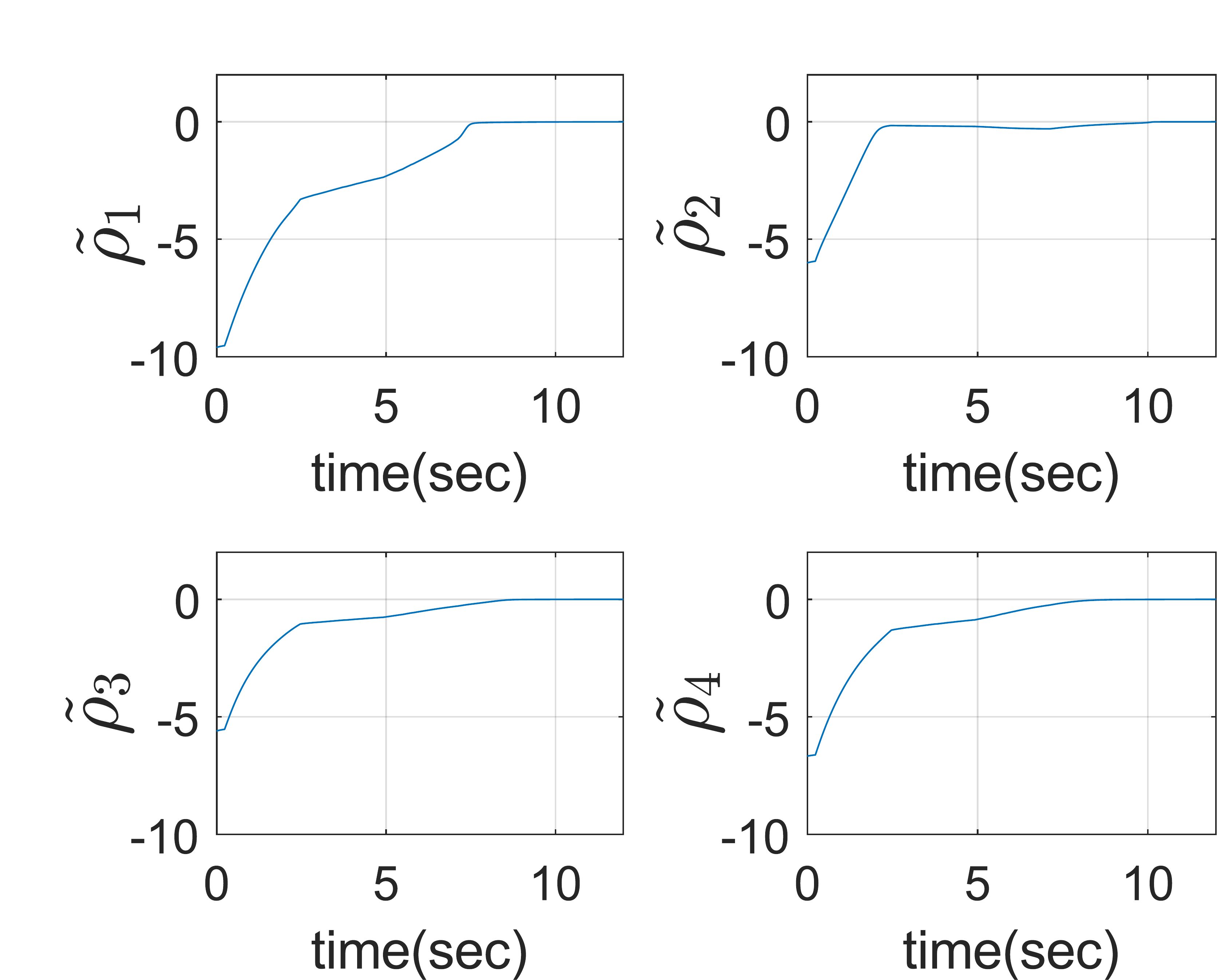}
  \caption{Estimation errors of the targets' positions.}
  \label{figure:Estimation errors of the targets' positions.}
\end{figure}

\medskip
  \emph{Case 2.} For the case of multiple agents, four agents are considered in this simulation. The initial positions of the four agents are $\y_1(0)=[8,0]^{\mathrm{T}}$, $\y_2(0)=[0,8]^{\mathrm{T}}$, $\y_3(0)=[-8,0]^{\mathrm{T}}$, $\y_4(0)=[0,-5]^{\mathrm{T}}$ respectively.
  Other parameter settings are shown in case 1.
  Trajectories of the agents are shown in \figurename~\ref{figure:agent}, which shows that multiple agents safely circumnavigate the minimum circle of multiple targets with the desired enclosing distance.
  The angular configuration is demonstrated by evolution of $\delta\phi_{i}$ for $i=1,2,3,4$.
  It can be seen from \figurename~\ref{figure:angle} that each $\delta\phi_{i}$ converges to $\pi/2$.
  One can observe that all the four agents converge to the desired minimum circle of the real targets and achieve even distribution over the minimum circle as expected.
  Meanwhile, the security of the agents is also guaranteed.

\begin{figure}[!htbp]
  \centerline{
  \subfigure[Trajectories of agents.]{
    \centering
    \includegraphics[height=4.1cm]{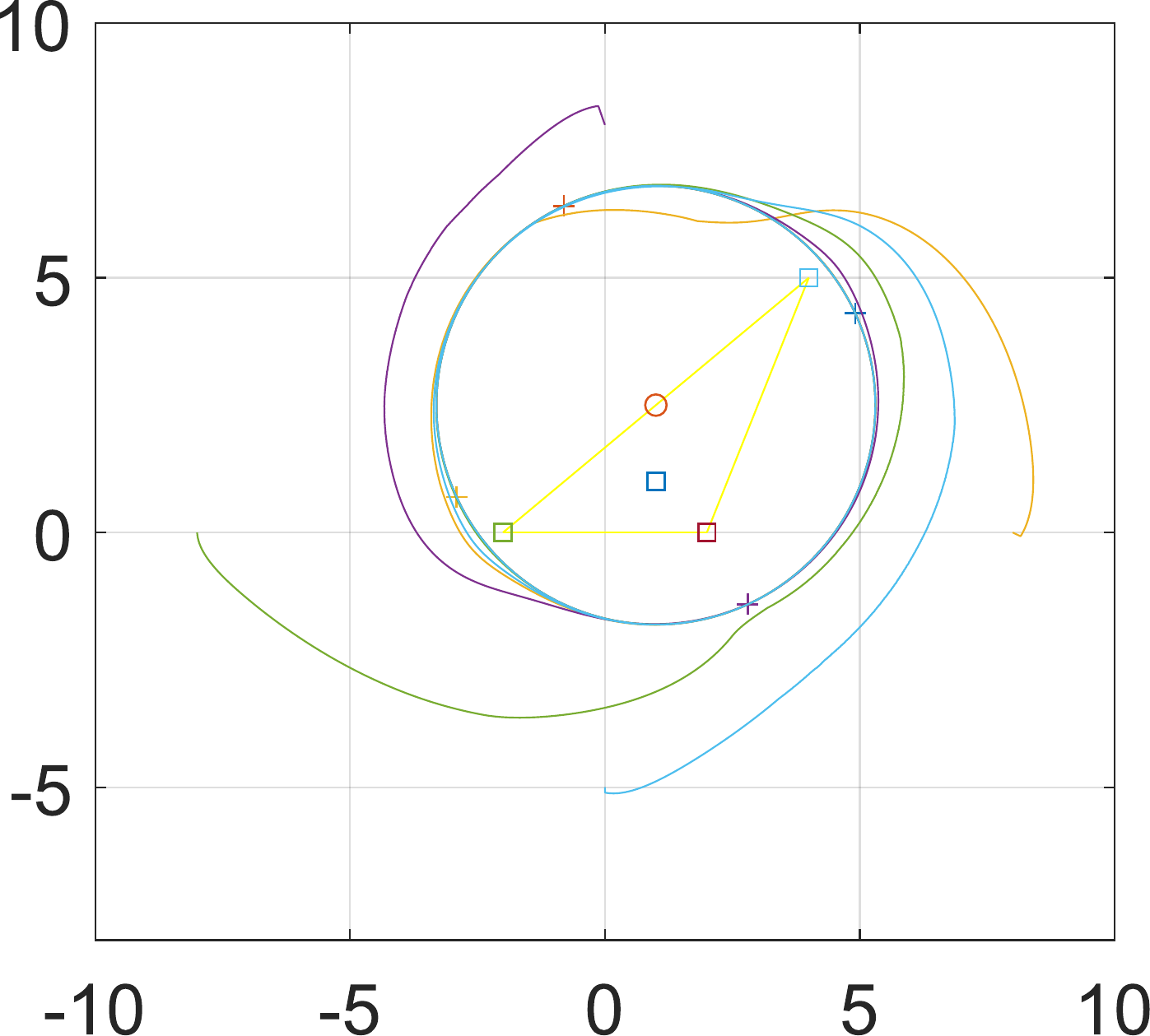}
    \label{figure:agent}
  }%
  \hspace{0.3cm}
  \subfigure[angle.]{
    \centering
    \includegraphics[height=4.1cm]{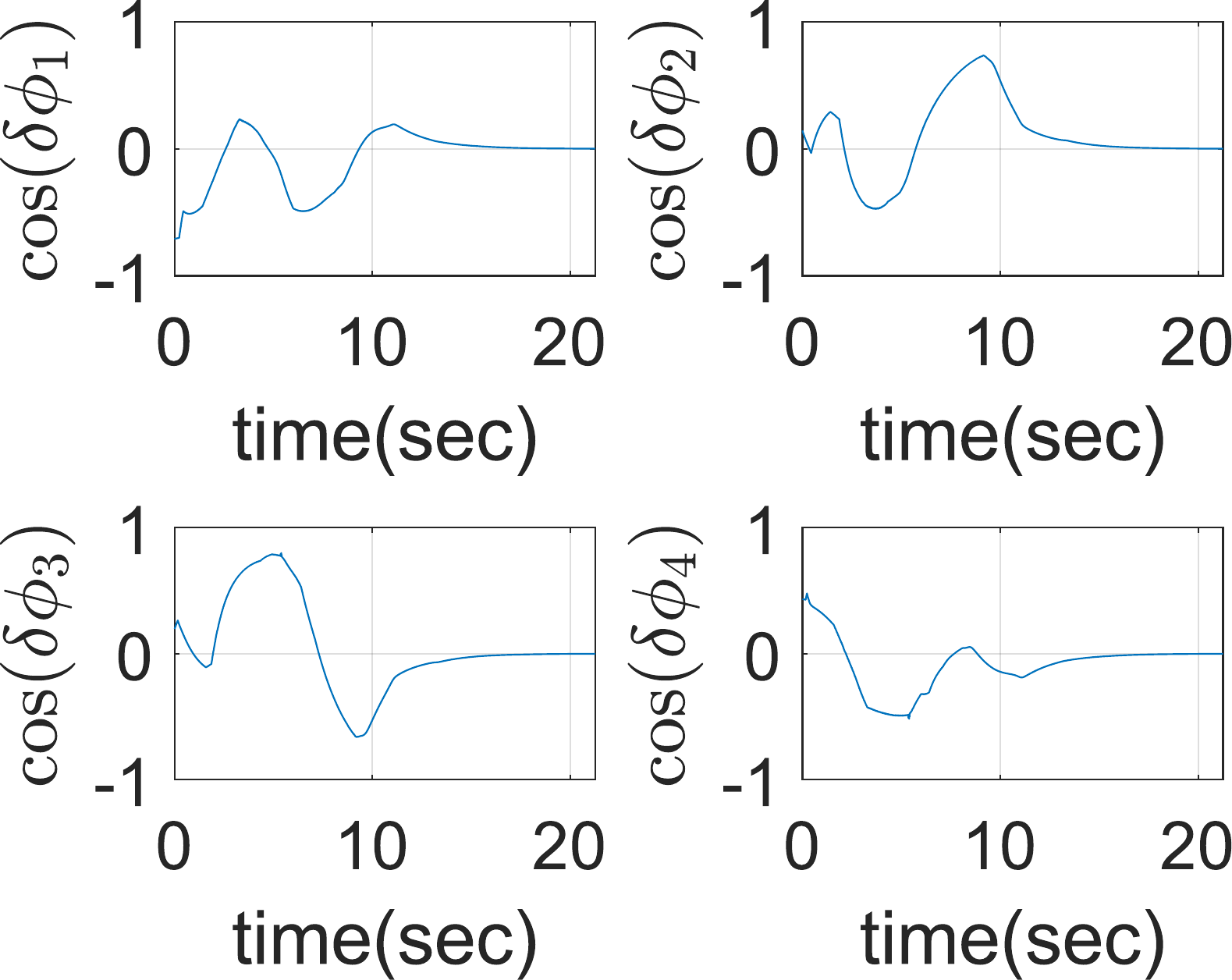}
    \label{figure:angle}
  } }
\caption{Entrapment of multiple targets by multiple agents in a minimum circle.}
\end{figure}

\medskip
  \emph{Case 3.} We further consider the case that multiple agents distributed in different radius with equiangular spaced formation. It can be seen in \figurename~\ref{figure:agent1} and \figurename~\ref{figure:angle1} that the four agents still converge to the desired minimum circle of the real targets and achieve even distribution over the different minimum circles as expected. Meanwhile, the security of the agents is also guaranteed.

\begin{figure}[!htb]
  \centerline{
  \subfigure[Trajectories of agents.]{
    \centering
    \includegraphics[height=4.1cm]{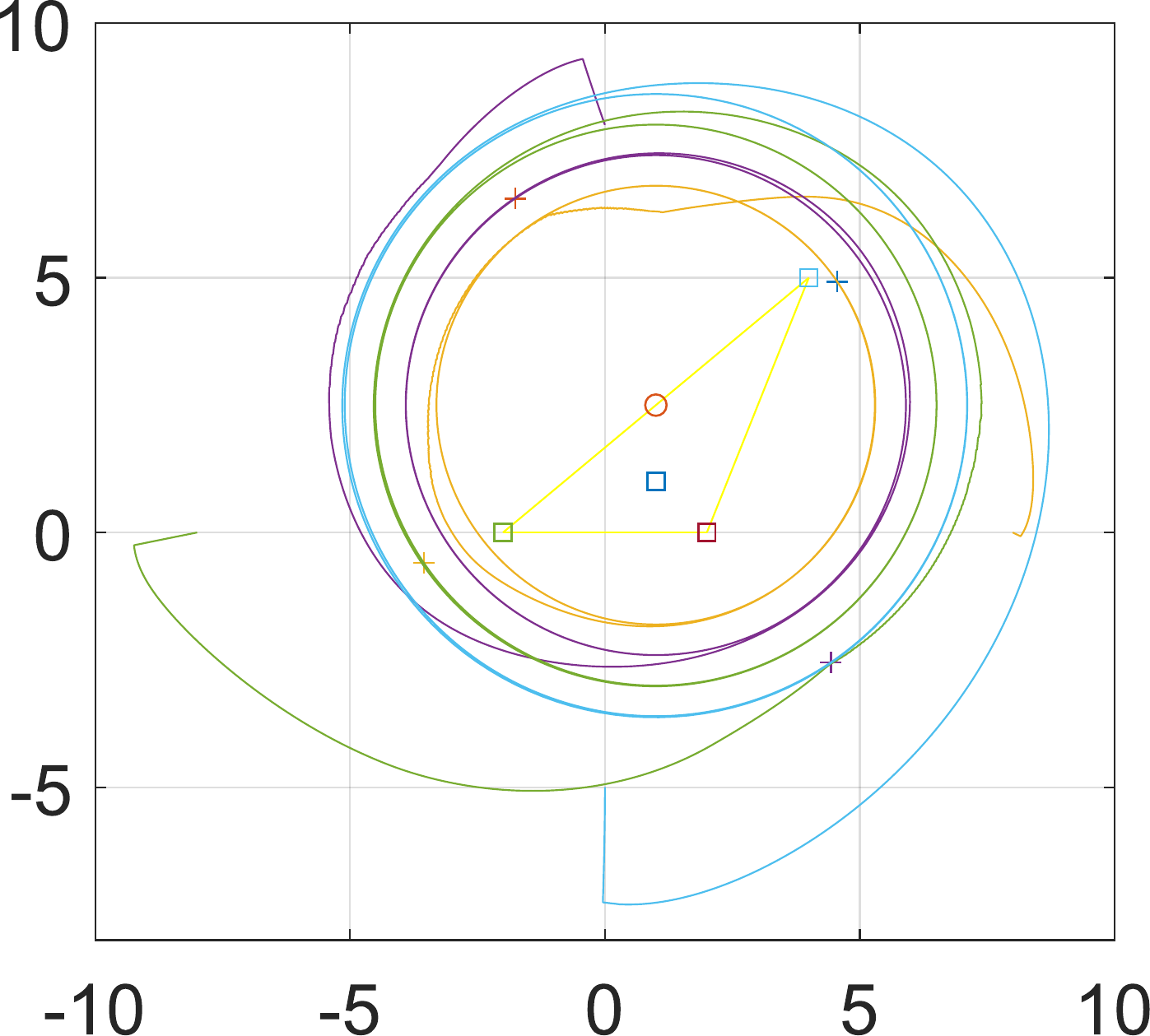}
    \label{figure:agent1}
  }%
  \hspace{0.3cm}
  \subfigure[angle.]{
    \centering
    \includegraphics[height=4.1cm]{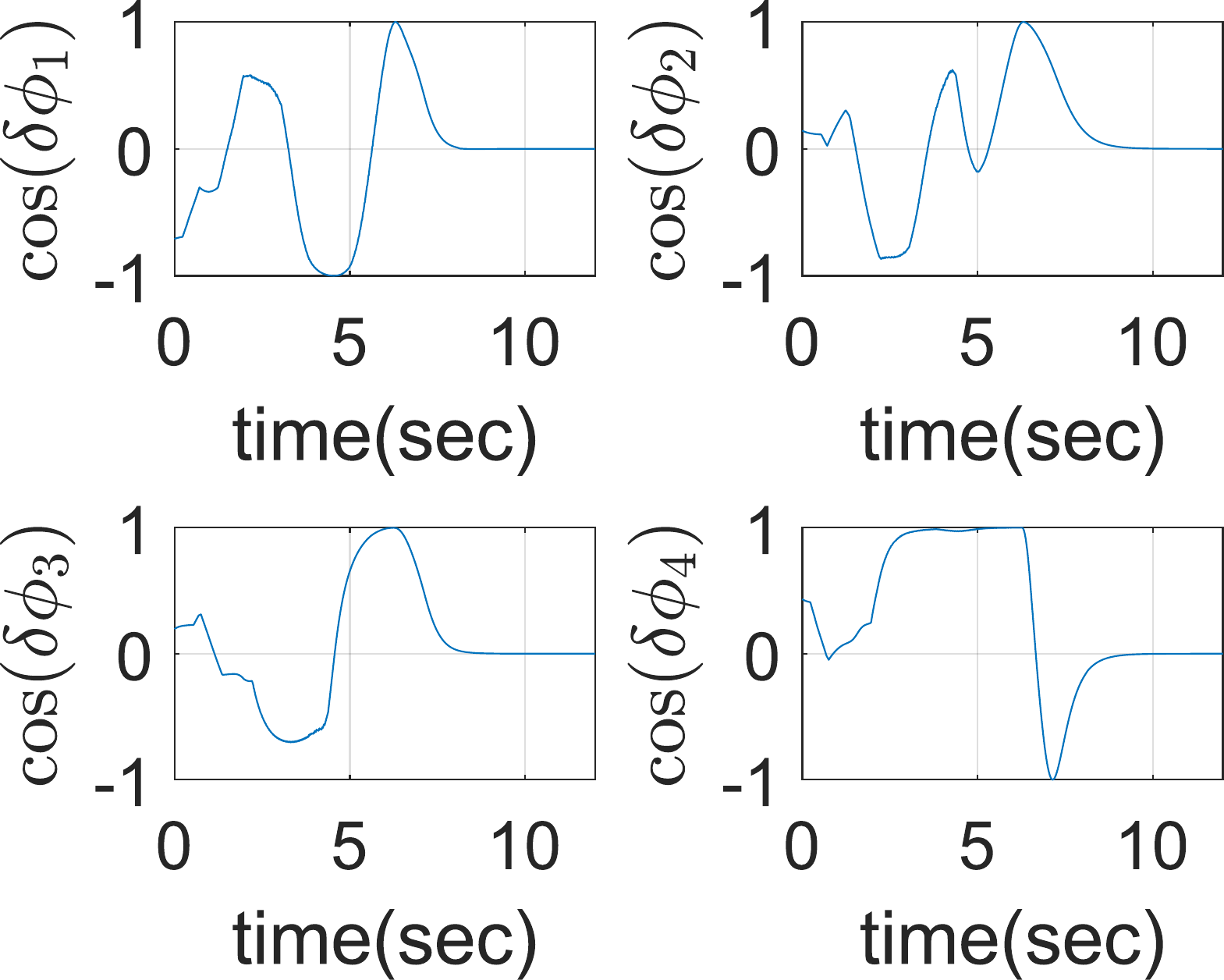}
    \label{figure:angle1}
  } }
\caption{Entrapment of multiple targets by multiple agents in different minimum circle.}
\end{figure}

\section{Conclusion and future works}\label{section:conclusion}

This paper studies safe minimum circle circumnavigation for a group of stationary targets based on a dynamic output feedback approach. A new circumnavigation algorithm is developed with the dynamic compensators and the control protocol keeping the radial motion far away from the dangerous area when the agent moving in a dangerous area. By using the proposed algorithm, the agent encloses the targets closely and does not collide with the multiple targets. Our results show that (\romannumeral1) the design of the control protocol facilitates the design of control law; (\romannumeral2) the prior information of the targets distribution range is not required in the design of the circumnavigation algorithm; and (\romannumeral3) this algorithm still work even in the multi-agent scenario.
Future investigations include improving the algorithm to fit in three-dimensional space scenarios and considering finite-time minimum circle circumnavigation problem.

\section*{Acknowledgements}
This work is supported in part by the National Natural Science Foundation of China under grant (No. 61973055) and the Fundamental Research Funds for the Central Universities (No. ZYGX2019J062).

\end{document}